\documentclass[pagebackref]{article}
\usepackage[a4paper]{geometry}

\usepackage[utf8]{inputenc}
\usepackage[T1]{fontenc}
\usepackage{xcolor}
\usepackage{amssymb,amsmath}

  \usepackage{amsthm}

\usepackage{dsfont} %
\usepackage{graphicx}
\usepackage{booktabs}
\usepackage{tabularx}
\usepackage{multicol}
\usepackage[ruled,vlined,linesnumbered]{algorithm2e}

\usepackage{tikz}
\usetikzlibrary{calc,shapes,positioning}

\usepackage{mathtools}

\usepackage{xargs} 
\newcommandx{\set}[2][1=1]{\ensuremath{\{#1,\ldots,#2\}}}
\newcommandx{\tlog}[3][1=,3=]{\log_{#1}^{#3}(#2)}

\usepackage[numbers,sort&compress]{natbib} 
\usepackage[linkcolor=red!50!black,citecolor=green!40!black,colorlinks]{hyperref}
\usepackage[capitalize,nameinlink]{cleveref}

  \newtheorem{theorem}{Theorem}
  \newtheorem{lemma}{Lemma}
  \newtheorem{proposition}{Proposition}
  \newtheorem{corollary}{Corollary}
  \newtheorem{observation}{Observation}
  \newtheorem{fact}{Fact}
  \theoremstyle{definition}
  \newtheorem{definition}{Definition}
  \newtheorem{problem}{Problem}

  \newtheorem{construction}{Construction}
\crefname{observation}{Observation}{Observations}
\crefname{rrule}{Reduction Rule}{Reduction Rules}
\crefname{construction}{Construction}{Constructions}
\Crefname{proposition}{Prop.}{Props.}
\crefname{proposition}{Proposition}{Propositions}
\Crefname{theorem}{Thm.}{Thm.}
\crefname{theorem}{Theorem}{Theorems}
\Crefname{corollary}{Cor.}{Cors.}
\crefname{corollary}{Corollary}{Corollaries}
\Crefname{fact}{Fact}{Facts}
\crefname{fact}{Fact}{Facts}
\crefname{figure}{Figure}{Figures}

\newcommand{\I}{\mathcal{I}}
\newcommand{\yes}{\texttt{yes}}

\newcommand{\RD}{$(\Rightarrow)\quad$}
\newcommand{\LD}{$(\Leftarrow)\quad$}
\newcommand{\tss}[1]{\textsuperscript{#1}}

\crefname{problem}{Problem}{Problems}
\Crefname{problem}{Prob.}{Probs.}
\newcommandx{\decprob}[6][3=Input,5=Question]{
  \begingroup
  \par\noindent\nopagebreak[4]
  \begin{problem}\label{prob:#2}\vspace{-0.75em}\colorbox{gray!17!white}{\textsc{#1}\index{problem!#1}}\nopagebreak[4]\end{problem}\nopagebreak[4]\vspace{-0.7em}
  \par\noindent\hangindent=\parindent\textbf{#3}:  #4\nopagebreak[4]
  \par\noindent\hangindent=\parindent\textbf{#5}:  #6
  \par\bigskip
  \endgroup
}

\newcommand{\N}{\mathbb{N}}
\newcommand{\Nzero}{\mathbb{N}_0}

\renewcommand{\O}{\mathcal{O}}

\newcommand{\prob}[1]{\textnormal{\textsc{#1}}}
\newcommand{\fvsTsc}{\prob{Feedback Vertex Set}}
\newcommand{\fvsAcr}{\prob{FVS}}
\newcommand{\tcTsc}{\prob{3-Coloring}}

\newcommand{\isTsc}{\prob{Independent Set}}

\newcommand{\cocl}[1]{\ensuremath{\operatorname{#1}}}

\newcommand{\NP}{\cocl{NP}}
\newcommand{\cqed}{\hfill$\diamond$}

\newcommand{\tref}[1]{\tss{(\Cref{#1})}}

\newcommand{\setto}{\ensuremath{\leftarrow}}

\newcommand{\eps}{\ensuremath{\varepsilon}}

\newcommand{\ceq}{\ensuremath{\coloneqq}}

\newcommand{\lqed}{}

\newcommand{\ceil}[1]{\lceil#1\rceil}

\definecolor{lilla}{HTML}{750787}

\usepackage{etoolbox}

\newcommand{\ExternalLink}{%
  \tikz[x=1.2ex, y=1.2ex, baseline=-0.05ex]{%
    \begin{scope}[x=1ex, y=1ex]
        \clip (-0.1,-0.1) 
            --++ (-0, 1.2) 
            --++ (0.6, 0) 
            --++ (0, -0.6) 
            --++ (0.6, 0) 
            --++ (0, -1);
        \path[draw, 
            line width = 0.5, 
            rounded corners=0.5] 
            (0,0) rectangle (1,1);
    \end{scope}
    \path[draw, line width = 0.5] (0.5, 0.5) 
        -- (1, 1);
    \path[draw, line width = 0.5] (0.6, 1) 
        -- (1, 1) -- (1, 0.6);
  }
}

\newcommand{\cfko}[3]{
    \node (c#1x) at (0+#2*\xr,0+#3*\yr)[xnode]{};
    \node (a#1x1) at (-0.5*\xr+#2*\xr,0.5*\yr+#3*\yr)[xnode]{};
    \node (a#1x2) at (0.5*\xr+#2*\xr,0.5*\yr+#3*\yr)[xnode]{};
    \node (a#1x3) at (0.5*\xr+#2*\xr,-0.5*\yr+#3*\yr)[xnode]{};
    \node (a#1x4) at (-0.5*\xr+#2*\xr,-0.5*\yr+#3*\yr)[xnode]{};
    \foreach \x in {1,...,4}{\draw[xedge] (c#1x) -- (a#1x\x);}
    }
\newcommand{\cfkoE}[1]{
\draw[xedge] (a#1x1) -- (a#1x2) -- (a#1x3) -- (a#1x4) -- (a#1x1);
}
\newcommand{\theR}[3]{%
  \cfko{#1}{#2}{#3};
  \cfkoE{#1};
  \node (xp#1x) at (-1*\xr+#2*\xr,0+#3*\yr)[xnode]{};
  \node (yp#1x) at (1*\xr+#2*\xr,0+#3*\yr)[xnode]{};
  \draw[xedge] (a#1x1) -- (xp#1x) -- (a#1x4);
  \draw[xedge] (a#1x2) -- (yp#1x) -- (a#1x3);
  \draw[xedge] (xp#1x) to [out=90,in=90,looseness=1.4](yp#1x);
  \node (x#1x) at (-1.5*\xr+#2*\xr,0+#3*\yr)[xnode]{};
  \node (y#1x) at (1.5*\xr+#2*\xr,0+#3*\yr)[xnode]{};
  \draw[xedge] (x#1x) -- (xp#1x);
  \draw[xedge] (y#1x) -- (yp#1x);
}

\newcommand{\theRcut}[3]{%
  \cfko{#1}{#2}{#3};
  \cfkoE{#1};
  \node (xp#1x) at (-1*\xr+#2*\xr,0+#3*\yr)[xnode]{};
  \node (yp#1x) at (1*\xr+#2*\xr,0+#3*\yr)[xnode]{};
  \draw[xedge] (a#1x1) -- (xp#1x) -- (a#1x4);
  \draw[xedge] (a#1x2) -- (yp#1x) -- (a#1x3);
  \draw[xedge] (xp#1x) to [out=90,in=90,looseness=1.4](yp#1x);
}

\newcommand{\theL}[4]{%
    \cfko{#1}{#3}{#4};
    \cfkoE{#1};

    \cfko{#2}{1.5+#3}{#4};
    \cfkoE{#2};

    \draw[xedge] (a#2x1) -- (a#1x2);
    \draw[xedge] (a#2x4) -- (a#1x3);
    \node (x#1x) at (-1*\xr+#3*\xr,0+#4*\yr)[xnode]{};
    \node (y#2x) at (2.5*\xr+#3*\xr,0+#4*\yr)[xnode]{};
    \draw[xedge] (a#1x1) -- (x#1x) -- (a#1x4);
    \draw[xedge] (a#2x2) -- (y#2x) -- (a#2x3);
}

\newcommand{\theLpath}[2]{%
  \draw[xpath] (x#1x) to (a#1x1) to (a#1x4) to (c#1x) to (a#1x2) to (a#1x3) to (a#2x4) to (a#2x1) to (c#2x) to (a#2x3) to (a#2x2) to (y#2x);
}

\newcommand{\tikzpramble}{%
  \def\teps{0.33}
  \tikzstyle{xnode}=[circle,scale=1/2,draw,fill=white];
  \tikzstyle{xnodex}=[circle,fill,scale=1/2,draw];
  \tikzstyle{xnodey}=[diamond,fill,scale=1/2,draw];
  \tikzstyle{xedge}=[thick,-];
  \tikzstyle{xedgedot}=[thick,-,dotted];
  \tikzstyle{xpath}=[color=blue,opacity=0.2,line cap=round,line width=5pt];
  \tikzstyle{xpathx}=[color=magenta,opacity=0.2,line cap=round,line width=5pt];
  \tikzstyle{xpathy}=[color=cyan,opacity=0.2,line cap=round,line width=5pt];
  \tikzstyle{xhili}=[circle,scale=1.25,opacity=0.25,fill,color=orange,draw];
  \tikzstyle{xhiliIS}=[circle,scale=1.25,opacity=0.25,fill,color=magenta,draw];
  \tikzstyle{xxhili}=[scale=1.25,opacity=0.25,fill,color=cyan,draw];
}

\newcommand{\theD}[3]{
      \node (x#1x1) at (-1.75*\xr+#2*\xr,-0.5*\yr+#3*\yr)[xnode]{};
      \node (y#1x1) at (1.75*\xr+#2*\xr,-0.5*\yr+#3*\yr)[xnode]{};
      \node (a#1x1) at (-1.25*\xr+#2*\xr,0+#3*\yr)[xnode]{};
      \node (a#1x2) at (0*\xr+#2*\xr,2*\yr+#3*\yr)[xnode]{};
      \node (a#1x3) at (1.25*\xr+#2*\xr,0*\yr+#3*\yr)[xnode]{};
      \draw[xedge] (x#1x1) -- (a#1x1);
      \draw[xedge] (a#1x1) to [out=90,in=180](a#1x2);
      \draw[xedge] (a#1x2) to [out=0,in=90](a#1x3);
      \draw[xedge] (a#1x3) -- (y#1x1);
      
      \node (c#1x1) at (0*\xr+#2*\xr,0*\yr+#3*\yr)[xnode]{};
      \node (c#1x2) at (-0.75*\xr+#2*\xr,0.5*\yr+#3*\yr)[xnode]{};
      \node (c#1x3) at (-0.75*\xr+#2*\xr,1*\yr+#3*\yr)[xnode]{};
      \node (c#1x4) at (-0.0*\xr+#2*\xr,1.5*\yr+#3*\yr)[xnode]{};
      \node (c#1x5) at (0.75*\xr+#2*\xr,1*\yr+#3*\yr)[xnode]{};
      \node (c#1x6) at (0.75*\xr+#2*\xr,0.5*\yr+#3*\yr)[xnode]{};
      
      \draw[xedge] (a#1x1) -- (c#1x1) -- (c#1x2) -- (c#1x3)-- (c#1x4)-- (c#1x5)-- (c#1x6)-- (c#1x1)-- (a#1x3);
      
      \node (b#1x1) at (0+#2*\xr,0.5*\yr+#3*\yr)[xnode]{};
      \node (b#1x2) at (-0.25+#2*\xr,1*\yr+#3*\yr)[xnode]{};
      \node (b#1x3) at (0.25+#2*\xr,1*\yr+#3*\yr)[xnode]{};
      \draw[xedge] (b#1x1) -- (b#1x2) -- (b#1x3) -- (b#1x1);
      
      \foreach\x in {c#1x2,c#1x3}{\draw[xedge] (a#1x1) -- (\x);}
      \foreach\x in {c#1x3,c#1x4,c#1x5}{\draw[xedge] (a#1x2) -- (\x);}
      \foreach\x in {c#1x6,c#1x5}{\draw[xedge] (a#1x3) -- (\x);}
      
      \foreach\x in {c#1x1,c#1x2,c#1x6}{\draw[xedge] (b#1x1) -- (\x);}
      \foreach\x in {c#1x2,c#1x3,c#1x4}{\draw[xedge] (b#1x2) -- (\x);}
      \foreach\x in {c#1x4,c#1x5,c#1x6}{\draw[xedge] (b#1x3) -- (\x);}
}

\newcommand{\theDpath}[1]{%
  \draw[xpath] (x#1x1) to (a#1x1) to [out=90,in=180](a#1x2) to (c#1x4) to (c#1x3) to (b#1x2) to (b#1x3) to (c#1x5) to (c#1x6) to (b#1x1) to (c#1x2) to (c#1x1) to (a#1x3) to (y#1x1);
}

\newcommand{\Kk}[5]{%
  \def\xsc{1.25}
  \node (a#1x1) at (0+#2*\xr*\xsc+#4*\xr,-0*\yr*\xsc+#3*\yr*\xsc+#5*\yr)[xnode]{};
  \node (a#1x2) at (-0.66*\xr*\xsc+#2*\xr*\xsc+#4*\xr,-0.4*\yr*\xsc+#3*\yr*\xsc+#5*\yr)[xnode]{};
  \node (a#1x3) at (0.66*\xr*\xsc+#2*\xr*\xsc+#4*\xr,-0.6*\yr*\xsc+#3*\yr*\xsc+#5*\yr)[xnode]{};
  \node (a#1x4) at (-0.3*\xr*\xsc+#2*\xr*\xsc+#4*\xr,-0.8*\yr*\xsc+#3*\yr*\xsc+#5*\yr)[xnode]{};
  \node (a#1x5) at (0.2*\xr*\xsc+#2*\xr*\xsc+#4*\xr,-1*\yr*\xsc+#3*\yr*\xsc+#5*\yr)[xnode]{};
  \foreach\x in{1,...,4}
  \foreach\y in{2,...,5}\draw[xedge,gray] (a#1x\x) -- (a#1x\y);
  \node at (0*\xr+#2*\xr*\xsc+#4*\xr,-0.5*\yr*\xsc+#3*\yr*\xsc+#5*\yr)[fill=white]{$K_p$};
}
    
\newcommand{\theYk}[3]{
      \Kk{#1A}{-0.9}{0}{#2}{#3};
      \Kk{#1B}{1}{0}{#2}{#3};
      \foreach\x in{1,...,5}\draw[xedge] (a#1Ax\x) -- (a#1Bx\x);
      
      \node (x#1xp) at (-1*\xr+#2*\xr,-1.75*\yr+#3*\xr)[xnode]{};
      \foreach\x in{1,...,5}\draw[xedge] (x#1xp) -- (a#1Ax\x);
      \node (x#1x) at (-1*\xr+#2*\xr,-2.25*\yr+#3*\xr)[xnode]{};
      \node (y#1xp) at (1*\xr+#2*\xr,-1.75*\yr+#3*\xr)[xnode]{};
      \foreach\x in{1,...,5}\draw[xedge] (y#1xp) -- (a#1Bx\x);
      \node (y#1x) at (1*\xr+#2*\xr,-2.25*\yr+#3*\xr)[xnode]{};
      \draw[xedge] (x#1xp) -- (x#1x);
      \draw[xedge] (y#1xp) -- (y#1x);
}

\newcommand{\theYkpath}[1]{
  \draw[xpath] (x#1x) to (x#1xp) to (a#1Ax3) to (a#1Ax5) to (a#1Ax4) to (a#1Ax2) to (a#1Ax1) to (a#1Bx1) to (a#1Bx3) to (a#1Bx5) to (a#1Bx4) to (a#1Bx2) to (y#1xp) to (y#1x);
}

\newcommand{\Grid}[7]{
  \foreach\x in {0,...,#2}{
    \foreach \y in {0,...,#3}{
      \node (#1\x\y) at (\x*\xr+#4*\xr,\y*\yr+#5*\yr)[xnode,#6]{};
    }
  }
  \pgfmathsetmacro\yx{int(#3 - 1)}
  \foreach \x in {0,...,#2}
    \foreach \y [count=\yi] in {0,...,\yx}  
      \draw[#7] (#1\x\y)--(#1\x\yi) ;
  \pgfmathsetmacro\yx{int(#2 - 1)}
  \foreach \x in {0,...,#3}
    \foreach \y [count=\yi] in {0,...,\yx}  
      \draw[#7] (#1\y\x)--(#1\yi\x) ;
}

\newcommand{\mytitle}{Feedback Vertex Set on Hamiltonian Graphs}

\makeatletter
  \def\abstractname{Abstract.}
  \renewenvironment{abstract}{%
      \if@twocolumn
        \section*{\abstractname}%
      \else
        \small
        \quotation
	\noindent{\bfseries\abstractname}%
      \fi}
      {\if@twocolumn\else\endquotation\fi}
\makeatother
\title{\Large \bf \mytitle}
\author{Dario~Cavallaro \and Till Fluschnik\footnote{Supported by DFG, project TORE (NI/369-18).}}
\date{\small Technische Universität Berlin, Faculty~IV,\\ Algorithmics and Computational Complexity, Germany.\\\texttt{cavallaro@campus.tu-berlin.de,till.fluschnik@tu-berlin.de}}

\begin{document}

\maketitle

\begin{abstract}
We study the computational complexity of 
\fvsTsc{} on subclasses of Hamiltonian graphs.
In particular,
we consider Hamiltonian graphs that are regular or
are planar and regular.
Moreover,
we study the less known class of~$p$-Hamiltonian-ordered graphs,
which are graphs that admit for any $p$-tuple of vertices
a Hamiltonian cycle visiting them in the order given by the tuple.
We prove that \fvsTsc{} remains \NP-hard in these restricted cases,
even if a Hamiltonian cycle is additionally given as part of the input.

\end{abstract}

\section{Introduction}

Hamiltonian graphs are graphs
admitting a cycle that visits every vertex
(exactly once).
We study the computational complexity 
of the following 
classic \NP-complete~\cite{Karp72}
problem on subclasses of Hamiltonian graphs.

\decprob{\fvsTsc{} (\fvsAcr{})}{fvs}
{An undirected graph~$G=(V,E)$ and an integer~$k\in\Nzero$.}
{Is there~$U\subseteq V$ with~$|U|\leq k$ such that~$G-U$ is acyclic?}

\noindent
We additionally restrict Hamiltonian graphs 
to be
planar 
(can be drawn on the two-dimensional plane 
with no two edges crossing except at their endpoints)
or
regular 
(every vertex has the same degree).
In particular,
we study the classes of 4-regular planar Hamiltonian graphs
and of 5-regular planar Hamiltonian graphs
(recall that there is no 6-regular planar graph).
Moreover,
we consider the class of~$p$-Hamiltonian-ordered graphs~\cite{NgS97}.
These Hamiltonian graphs admit for each $p$-tuple~$(x_1,\dots,x_p)$ of vertices
a Hamiltonian cycle that visits the vertices~$x_1,\dots,x_p$ in this order.
The class of $p$-Hamiltonian-ordered graphs form a subclass of $p$-ordered Hamiltonian graphs,
the latter being Hamiltonian graphs that for any~$p$-tuple~$(x_1,\dots,x_p)$
admit a cycle that visits~$x_1,\dots,x_p$ in this order.
The class of $p$-ordered Hamiltonian graphs form a subclass of~$(p-1)$-connected Hamiltonian graphs
(graphs in which every pair of vertices is connected via~$(p-1)$ internally vertex-disjoint paths).
Finally,
for \fvsTsc{} on these subclasses of Hamiltonian graphs,
we also study the more restricted case when a Hamiltonian cycle 
is additionally provided in the input
(recall that computing a Hamiltonian cycle is \NP-complete in general~\cite{Karp72}).

\paragraph*{Related Work.}

\isTsc{} 
remains \NP-complete on
3- and 4-regular Hamiltonian graphs~\cite{FleischnerSS10},
which enabled to prove \NP-hardness for a temporal graph problem with two layers~\cite{FluschnikNRZ19}.
\tcTsc{} remains \NP-complete on 4- and 5-regular planar graphs~\cite{Dailey80},
and on 4-regular Hamiltonian graphs~\cite{FleischnerS03}.
\fvsTsc{} remains \NP-complete on planar graphs of maximum degree four~\cite{Speckenmeyer88}
and
is polynomial-time solvable on 
maximum degree-three~\cite{UenoKG88} and
3-regular graphs~\cite{LiL99},
chordal graphs,
permutation graphs,
split graphs~\cite{FestaPR99}.

\paragraph*{Our Contributions.}
\cref{fig:results} gives an overview of our results.
\begin{figure}[t!]
 \centering
  \begin{tikzpicture}
 
    \usetikzlibrary{patterns,backgrounds,shapes}
    \usetikzlibrary{calc}

      \def\yr{1.1}
      \def\xr{1.05}

      \def\boxw{2.9}
    \def\boxh{1.75}

    \def\colNP{orange!50!red!18!white}
    \def\colP{green!20!white}
    \def\colOpen{blue!20!white}
    
      \def\fsres{\footnotesize}
      \def\fsresx{\scriptsize}
    
    \def\bwA{1}

    \tikzstyle{xarc}=[->,gray,thick,>=latex,rounded corners]
    
    \newcommand{\gbox}[6]{
      \node (a#1) at (#2)[rectangle, rounded corners, minimum width=\xr*\boxw cm, minimum height=\yr*\boxh cm,fill=lightgray,very thick,draw]{};
      \node (axx#1) at (a#1.north)[anchor=north, rectangle, rounded corners, minimum width=\xr*\boxw cm, minimum height=0.33*\yr*\boxh cm,fill=white!95!black,draw]{{\small Hamiltonian}};
      \node (ax#1) at (axx#1.south)[anchor=north, rectangle, rounded corners, minimum width=\xr*\boxw cm, minimum height=0.33*\yr*\boxh cm,fill=white,very thick,draw]{{\small#3}};
      \node (bn#1) at (a#1.south west)[anchor=south west,rectangle, rounded corners, minimum width=\xr*\boxw*\bwA cm, minimum height=0.33*\yr*\boxh cm,align=center,font=\fsresx,fill=#5,text width= \xr*\boxw*\bwA*0.9 cm,draw]{#4};
      \draw[rounded corners,thick] (bn#1.south west) rectangle (a#1.north east);
    }

    \gbox{h}{0.0*\xr,0*\yr}{}{\NP-hard}{\colNP};

    \def\ysh{2.25}
    \def\xshA{1.675}
    \def\xshB{5}
    \gbox{conn}{-\xshB*\xr,-0.5*\ysh*\yr}{$(p-1)$-connected}{\NP-hard\tref{cor:ordconhard}}{\colNP};
    \gbox{3reg}{-\xshA*\xr,-\ysh*\yr}{$\leq3$-regular}{Poly-time\tss{\cite{UenoKG88,LiL99}}}{\colP};
    \gbox{reg}{\xshA*\xr,-\ysh*\yr}{$\geq 4$-regular}{\NP-hard\tref{thm:fvs:kreg}}{\colNP};
    \gbox{ord}{-\xshB*\xr,-1.5*\ysh*\yr}{$p$-ordered}{\NP-hard\tref{cor:ordconhard}}{\colNP};
    \gbox{plan}{\xshB*\xr,-0.5*\ysh*\yr}{planar}{\NP-hard}{\colNP};
    \gbox{ordham}{-\xshB*\xr,-2.5*\ysh*\yr}{$p$-Hamilt.-ordered}{\NP-hard\tref{thm:fvs:pho}}{\colNP};
    \gbox{plan3reg}{-\xshA*\xr,-2.5*\ysh*\yr}{planar $\leq3$-regular}{Poly-time\tss{\cite{UenoKG88,LiL99}}}{\colP};
    \gbox{plan5reg}{\xshB*\xr,-2.5*\ysh*\yr}{planar 5-regular}{\NP-hard\tref{thm:fvs:5regplanham}}{\colNP};
    \gbox{plan4reg}{\xshA*\xr,-2.5*\ysh*\yr}{planar 4-regular}{\NP-hard\tref{thm:fvs:4regplanaHam}}{\colNP};
    \draw[xarc] (aord) to node[midway,right]{\cite{NgS97}}(aconn);
    \draw[xarc] (aordham) to node[midway,right]{\Cref{fact:phopoh}}(aord);
    \draw[xarc] (aplan3reg) to [out=90,in=-90](aplan);
    \draw[xarc] (aplan3reg) to (a3reg);
    \draw[xarc] (aplan4reg) to (aplan);
    \draw[xarc] (aplan4reg) to (areg);
    \draw[xarc] (aplan5reg) to (aplan);%
    \draw[xarc] (aplan5reg) to (areg);%
    
    \draw[xarc] (a3reg) to (ah);
    \draw[xarc] (areg) to (ah);
    \draw[xarc] (aconn) to (ah);
    \draw[xarc] (aplan) to
    (ah);

  \end{tikzpicture}
  \caption{Overview of our results.
  In each box,
  the lowest level describes the computational complexity 
  (\NP-hard versus polynomial-time) 
  of \fvsTsc{} on the graph class described by the two upper layers.
  An arrow from a box~$A$ to a box~$B$ describes that~$A$'s graph class is included in $B$'s graph class.
  All shown \NP-hardness results hold true
  even if a Hamiltonian cycle is provided as part of the input.
  }
  \label{fig:results}
\end{figure}
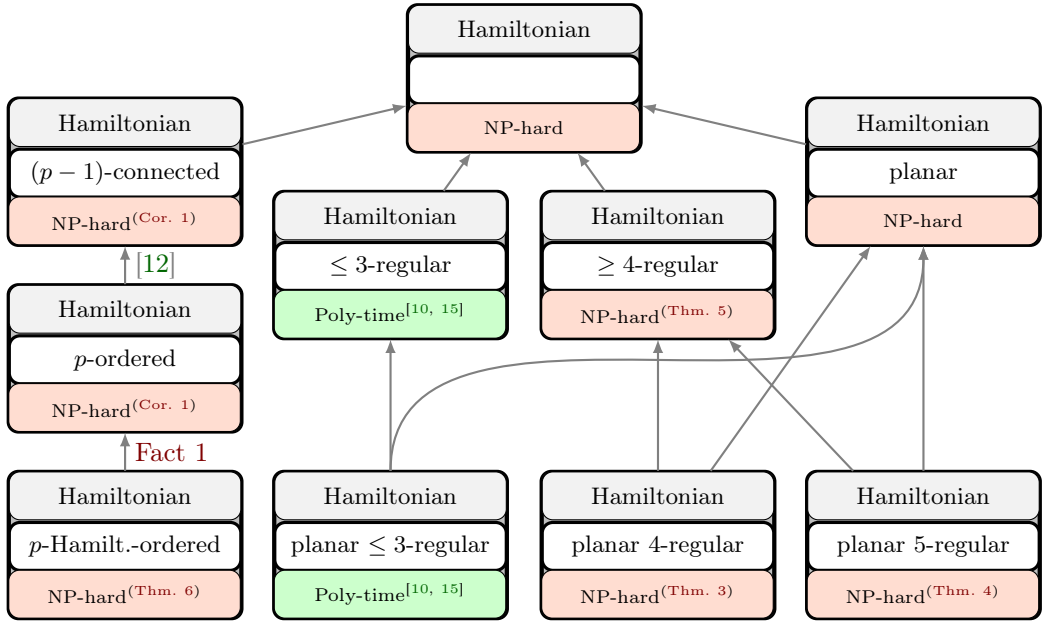
We prove that \fvsTsc{} is
\NP-hard on 4- and 5-regular planar Hamiltonian graphs
as well as on~$p$-regular Hamiltonian graphs for every~$p\geq 4$.
Moreover,
we prove that \fvsTsc{} is \NP-hard on $p$-Hamiltonian-ordered graphs for every~$p\geq 3$,
which implies \NP-hardness on~$p$-ordered Hamiltonian graphs
and further \NP-hardness on~$(p-1)$-connected Hamiltonian graphs.
Finally,
all our \NP-hardness results still hold true
if a Hamiltonian cycle is additionally provided as part of the input.

\section{Preliminaries}
\label{sec:prelims}

We denote by~$\N$ and~$\Nzero$ the natural numbers excluding and including zero,
respectively.
We use basic notations from graph theory~\cite{Diestel10,balakrishnan2012textbook}.

  \paragraph*{Graph Theory.}

For two graphs~$G,H$,
we denote by~$G*H$ the graph with vertex set~$V(G)\cup V(H)$
and edge set~$E(G)\cup E(H)\cup \{\{v,w\}\mid v\in V(G),w\in V(G)\}$.
We denote by~$K_n$ the complete graph on~$n\in\N$ vertices.
We denote by~$C_n$ the cycle on~$n\in\N$ vertices.
The neighborhood~$N_G(v)$ of a vertex~$v\in V$ in~$G$
is the vertex set~$\{w\in V\mid \{v,w\}\in E\}$.
Let~$v,w$ be two distinct vertices in~$G=(V,E)$.
The graph obtained by \emph{identifying~$v$ with~$w$}
has vertex set~$(V\setminus\{v,w\})\cup\{vw\}$,
where~$vw$ is a new vertex,
and
edge set~$(E\setminus\{e\in E\mid \{v,w\}\cap e\neq\emptyset\})\cup\{\{vw,x\}\mid x\in (N_G(v)\cup N_G(w))\setminus\{v,w\}\}$.

  \paragraph*{Hamiltonian Graphs and Subclasses.}

A graph~$G$ is~$p$-ordered
if for every sequence~$v_1,\dots,v_p$ of distinct vertices of~$G$ 
there exists a cycle~$C$ in~$G$ 
that encounters the vertices~$v_1,\dots,v_p$ in this order.
A graph~$G$ is called~$p$-Hamiltonian-ordered
if for every sequence~$v_1,\dots,v_p$ of distinct vertices of~$G$ there
exists an Hamiltonian cycle~$C$ 
that encounters the vertices~$v_1,\dots,v_p$ in this order.
Clearly:

\begin{fact}
 \label{fact:phopoh}
 Every $p$-Hamiltonian-ordered graph is~$p$-ordered Hamiltonian.
\end{fact}

\paragraph*{Graph tool box.}
We use several graphs as gadgets 
for our \NP-hardness reductions
that we collect in this ``graph tool box'' 
(see~\cref{fig:gtb}).
\begin{figure}[t]
 \centering
  \begin{tikzpicture}

    \def\xr{1}
    \def\yr{1}
    \tikzpramble{};
    
    \begin{scope}
      \theR{1}{0}{0};
      \node at (x1x)[xnode,label=90:{$x$}]{};
      \node at (y1x)[xnode,label=90:{$y$}]{};
      \node at (xp1x)[xnode,label=-90:{$x'$}]{};
      \node at (yp1x)[xnode,label=-90:{$y'$}]{};
      
      \node at (-1.75*\xr,1*\yr)[]{(a)}; 
      \draw[xpath] (x1x) to (xp1x) to (a1x1) to (a1x4) to (c1x) to (a1x2) to (a1x3) to (yp1x) to (y1x);
      \node at (xp1x)[xhili]{};
      \node at (c1x)[xhili]{};
      \node at (a1x3)[xhili]{};
      \node at (0,-1*\yr)[]{Graph~$R$};
    \end{scope}

    \begin{scope}[xshift=5*\xr cm]
    
    \theL{1}{2}{0}{0}
    \node at (x1x)[label=90:{$x$}]{};
    \node at (y2x)[label=90:{$y$}]{};
    \node at (-1.25*\xr,1*\yr)[]{(b)}; 
    \theLpath{1}{2}
    \node at (a1x1)[xhili]{};
    \node at (a1x3)[xhili]{};
    \node at (a2x1)[xhili]{};
    \node at (a2x3)[xhili]{};
    \node at (0.75*\xr,-1*\yr)[]{Graph~$L$};
    \end{scope}

    \begin{scope}[xshift=6.0*\xr cm,yshift=-1.75*\yr cm]
      \theYk{a}{0}{0}
      \node at (xaxp)[label=180:{$x'$}]{};
      \node at (xax)[label=180:{$x$}]{};
      \node at (yaxp)[label=180:{$y'$}]{};
      \node at (yax)[label=0:{$y$}]{};
      
      \node at (-2.25*\xr,0.25*\yr)[]{(d)}; 
      \node at (0.0*\xr,-2.75*\yr)[]{Graph~$Y_p$};
      \theYkpath{a}
      \node at (xaxp)[xhili]{};
      \node at (aaAx1)[xhili]{};
      \node at (aaAx2)[xhili]{};
      \node at (aaAx3)[xhili]{};
      \node at (yaxp)[xhili]{};
      \node at (aaBx4)[xhili]{};
      \node at (aaBx5)[xhili]{};
      \node at (aaBx1)[xhili]{};
    \end{scope}
    
    \begin{scope}[xshift=0.25*\xr cm,yshift=-3.75*\yr cm]

      \theD{1}{0}{0}
      \node at (x1x1)[label=90:{$x$}]{};
      \node at (y1x1)[label=90:{$y$}]{};
      \node at (a1x1)[label=135:{$x'$}]{};
      \node at (a1x3)[label=45:{$y'$}]{};
      \node at (a1x1)[xhili]{};
      \node at (a1x3)[xhili]{};
      \node at (b1x1)[xhili]{};
      \node at (c1x4)[xhili]{};
      \node at (b1x2)[xhili]{};
      \node at (c1x5)[xhili]{};
      \node at (-2*\xr,2.25*\yr)[]{(c)}; 
      \theDpath{1}
    \node at (0*\xr,-0.75*\yr)[]{Graph~$D$};
    \end{scope}
    \end{tikzpicture}
    \caption{Our graph tool box with (a) the graph~$R$,
    (b) the graph~$L$,
    (c) the graph~$D$,
    and~(d) the graph~$Y_p$.
    For each graph, 
    a Hamiltonian path (blue) 
    as well as a minimum feedback vertex set (orange) are depicted.}
    \label{fig:gtb}
\end{figure}
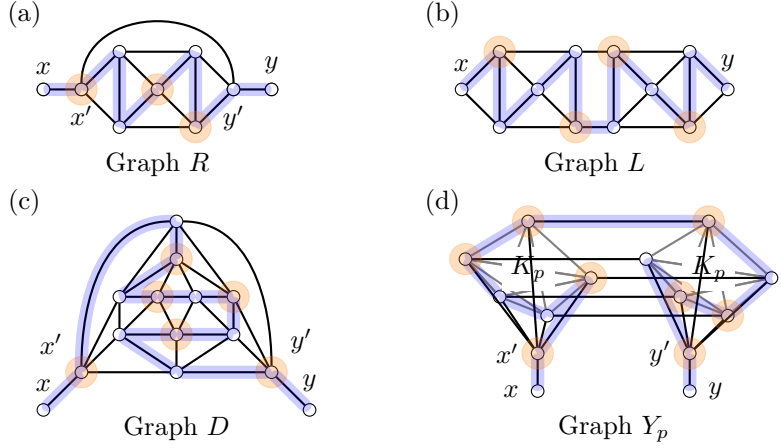

\subparagraph*{The Graph~$R$ (see~\cref{fig:gtb}(a)):}
 Let~$R$ denote the graph obtained from a~$C_4*K_1$ by adding two
 vertices~$x'$ and~$y'$ and making each adjacent with exactly two different vertices of degree three such that all vertices
 except for~$x'$ and~$y'$ have degree four.
 Add vertex~$x$ and make it adjacent with~$x'$,
 and add vertex~$y$ and make it adjacent with~$y'$.
 Finally, make~$x'$ adjacent with~$y'$.
We have the following simple yet useful observation on a~$C_4*K_1$:

\begin{observation}
 \label{obs:c4k1}
 The graph~$C_4*K_1$ admits no feedback vertex set of size one yet one of size two.
\end{observation}

\subparagraph*{The Graph~$L$ (see~\cref{fig:gtb}(b)):}
 Let~$L$ denote the graph obtained as follows.
 Take two disjoint~$C_4*K_1$s.
 Let~$\{v,w\}$ be an edge of one~$C_4*K_1$ with both~$v,w$ being of degree three,
 and~$\{v',w'\}$ analogously from the other~$C_4*K_1$.
 Make~$v$ adjacent with~$v'$ and~$w$ adjacent with~$w'$.
 Let~$\{x',x''\}$ and~$\{y',y''\}$ denote the two edges with vertices of degree three.
 Add a vertex~$x$ and make it adjacent with~$x',x''$,
 and add a vertex~$y$ and make it adjacent with~$y',y''$.

\noindent
\subparagraph*{The Graph~$D$ (see~\cref{fig:gtb}(c)):}
Let~$D$ denote the graph obtained as follows.
Take a~$C_6$, 
say with vertex set~$\{c_0,\dots,c_5\}$ and edge set~$\{\{c_i,c_{i+1\bmod 6}\}\mid i\in\set[0]{5}\}$.
Add a~$K_3$, 
say with vertex set~$\{v_1,v_2,v_3\}$.
Make~$v_1$ adjacent with~$c_0,c_1,c_5$,
$v_2$ adjacent with~$c_1,c_2,c_3$,
and~$v_3$ adjacent with~$c_3,c_4,c_5$.
Add a vertex~$z$ and make it adjacent with~$c_2,c_3,c_4$.
Add vertices~$x$ and~$x'$,
and make~$x'$ adjacent with~$x,z,c_0,c_1,c_2$.
Finally,
add vertices~$y$ and~$y'$,
and make~$y'$ adjacent with~$y,z,c_0,c_4,c_5$.

\subparagraph*{The Graphs~$Y_p$ (see~\cref{fig:gtb}(d)):}
 Let~$p\in\N$ with~$p\geq 3$.
 Let~$Y_p$ denote the graph obtained as follows.
 Take two disjoint~$A\ceq K_p$ and~$B\ceq K_p$.
 Add a perfect matching between the vertices of~$A$ and~$B$.
 Next, 
 add two vertices~$x,x'$ and make~$x'$ adjacent to all vertices in~$V(A)\cup\{x\}$.
 Finally,
 add two vertices~$y,y'$ and make~$y'$ adjacent to all vertices in~$V(B)\cup\{y\}$.

\begin{definition}[Insertion]
 Let~$G$ be a graph and
 $u,v\in V(G)$.
 An~$H$-insertion at~$u,v$ with~$H\in\{R,L,D\}\cup\bigcup_{p\geq 3}\{Y_p\}$ results in the graph obtained from~$G$
 by adding a copy of~$H$ to~$G$ 
 and identifying~$x$ with~$u$ 
 and~$y$ with~$v$.
\end{definition}

\section{Planar Regular Hamiltonian Graphs}
\label{sec:planreg}

In this section,
we prove that
\fvsTsc{} remains \NP-hard 
on 4-regular planar Hamiltonian graphs 
and 
on 5-regular planar Hamiltonian graphs,
in both cases
even if a Hamiltonian cycle is provided.
We first prove that \fvsAcr{} is \NP-hard on 4-regular planar graphs (\cref{ssec:4regplanar}),
then make the graph Hamiltonian (\cref{ssec:4regplanarham}),
and finally make the graph 5-regular (\cref{ssec:5regplanarham}).

\subsection{4-regular planar}
\label{ssec:4regplanar}

\fvsTsc{} is \NP-hard even on connected planar graphs of maximum degree four~\cite{SpeckenmeyerPHD,Speckenmeyer88}.
We strengthen this with the following.

\begin{theorem}%
 \label{thm:fvs:4regplanar}
 \fvsTsc{} 
 is \NP-hard
 on connected 4-regular planar graphs.
\end{theorem}

\noindent
We can delete degree-zero and -one vertices from a graph,
and obtain an equivalent instance.
Next,
we deal first with degree-two vertices
to prove that \fvsAcr{} 
is \NP-hard 
on planar graphs of minimum degree three and maximum degree four
(\cref{prop:fvs:threefour}),
and then we deal with the remaining vertices of degree three.

\subsubsection{Degree-two vertices}

We next make each degree-two vertex a degree-four vertex
to obtain the following.

\begin{proposition}%
 \label{prop:fvs:threefour}
 \fvsTsc{} 
 is \NP-hard
 on connected planar graphs of minimum degree three and maximum degree four.
\end{proposition}

\noindent
To prove~\cref{prop:fvs:threefour},
we will perform an~$R$-insertion on each degree-two vertex.
We have the following crucial observation on~$R$.

\begin{lemma}%
 \label{obs:R}
 Graph~$R$ is planar,
 admits a Hamiltonian~$x$-$y$~path,
 and has no feedback vertex set of size at most two,
 yet one of size three containing~$x'$ or~$y'$,
 but none of size three containing~$x$ or~$y$.
\end{lemma}

\begin{proof}
  That~$R$ is planar and admits a feedback vertex set of size three containing~$x'$ or~$y'$
  is depicted in~\cref{fig:gtb}.
  We need to delete two vertices from~$C_4*K_1$ (\cref{obs:c4k1}).
  Observe that deleting any two vertices from~$C_4*K_1$ leaves two vertex-disjoint paths between~$x'$ and~$y'$ in~$R$
  (one using the edge~$\{x',y'\}$,
  one passing through what remains of~$C_4*K_1$).
  Thus,
  deleting one of~$x'$ or~$y'$ is required.
  It follows that~$R$ admits no feedback vertex set of size at most two.
  \lqed
\end{proof}

\noindent
An immediate consequence of~\cref{obs:R} is the following.

\begin{observation}
 \label{obs:RinsPlus3}
 Let$~\I=(G,k)$ be an instance of~\fvsTsc{}
 and let~$u,v\in V(G)$.
 Let~$G'$ be the graph obtained from an~$R$-insertion at~$u,v$
 and let~$k'\ceq k+3$.
 Then~$\I$ is a \yes-instance
 if and only if
 $(G',k')$ is a \yes-instance of~\fvsTsc{}.
\end{observation}

\noindent
We are set to prove~\cref{prop:fvs:threefour}.

\begin{proof}[Proof of~\cref{prop:fvs:threefour}]
 Let~$\I=(G,k)$ be an instance of~\fvsTsc{}
 where~$G$ is a planar graph of minimum degree two and maximum degree four.
 Let~$d$ denote the number of degree-two vertices in~$G$.
 Let~$G'$ be the graph obtained from~$G$
 by applying for each degree-two vertex~$v$
 an~$R$-insertion at~$v,v$.
 Note that each of these~$d$~$R$-insertions preserves planarity and equivalence.
 Let~$\I'\ceq (G',k')$ be the obtained instance where~$k'\ceq k+3d$.
 Due to~\cref{obs:RinsPlus3},
 we know that~$\I$ is a \yes-instance 
 if and only if
 $\I'$ is a \yes-instance.
 \lqed
\end{proof}

\subsubsection{Degree-three vertices}

Next,
we deal with degree-three vertices.
We will employ
the following
specific
straight-line 
embedding of our graph on a grid.

\begin{theorem}[\cite{FraysseixPP90}]
  \label{thm:gridembedding}
  For any planar graph
  with $n$ vertices
  one can compute 
  in~$\O(n^2)$ time
  a straight-line embedding on the~$2n-4$ by~$n-2$ grid.
\end{theorem}

\noindent
We start with an embedding.
Let~$p(v)=(i,j)\in\set{2n-4}\times \set{n-2}$ 
be the coordinate of vertex~$v$ in the grid-embedding.
We aim for connecting the remaining 
degree-three vertices 
in a pairwise manner
(note that there is an even number of these).
To this end,
we construct chains of~$R$s connecting two degree-three vertices.
To ensure polynomial running time and planarity of the construction,
we need to identify the pairs of degree-three vertices which we want to connect
such that the ``$R$-chains'' are pairwise non-crossing and  vertex-disjoint.
To this end,
we apply a ``left-to-right bottom-to-top'' approach as follows
(see~\cref{fig:grid} for an illustration).
We iterate over vertices from left to right by coordinates,
that is,
by~$(i,\cdot)$ for increasing~$i$.
Thereby,
for each~$i$,
we iterate over~$(i,j)$ with increasing~$j$.
Once two vertices of degree three are discovered,
we connect them in a ``down-first out-most''-manner with an ``$R$-chain''.
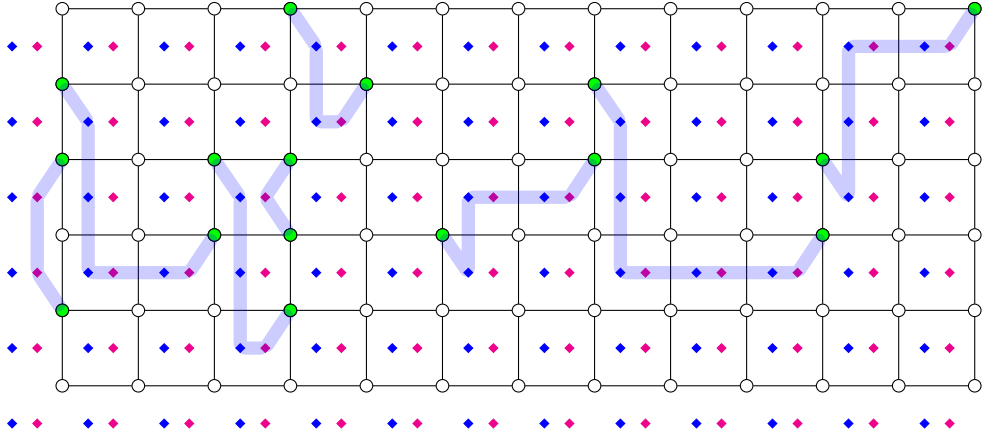
\begin{figure}[t]
 \centering
 \begin{tikzpicture}

    \def\xr{1}
    \def\yr{1}
  
  \tikzpramble{}
  \Grid{a}{12}{5}{0}{0}{}{};
  \Grid{b}{12}{5}{-0.66}{-0.5}{xnodey,color=blue}{xedgedot,color=blue,draw=none};
  \Grid{c}{12}{5}{-0.33}{-0.5}{xnodey,color=magenta}{xedgedot,color=blue,draw=none};
  
  \newcommand{\markV}[1]{%
    \node at (#1)[xnode,fill=green]{};
  }
  \markV{a01}
  \markV{a03}
  \markV{a04}
  \draw[xpath] (a01) to (c02) to (c03) to (a03);
  
  \markV{a22}
  \draw[xpath] (a04) to (b14) to (b12) to (c22) to (a22);
  
  \markV{a23}
  
  \markV{a35}
  \markV{a33}
  \markV{a32}
  \markV{a31}
  
  \markV{a44}
  
  \markV{a52}
  \markV{a73}
  
  \markV{a74}
  \markV{a102}
  
  \markV{a103}
  \markV{a125}

  \draw[xpath] (a23) to (b33) to (b31) to (c31) to (a31);
  \draw[xpath] (a32) to (c33) to (a33);
  \draw[xpath] (a35) to (b45) to (b44) to (c44) to (a44);
  \draw[xpath] (a52) to (b62) to (b63) to (c73) to (a73);
  \draw[xpath] (a74) to (b84) to (b82) to (c102) to (a102);
  \draw[xpath] (a103) to (b113) to (b115) to (c125) to (a125);

 \end{tikzpicture} 
 \caption{Illustration to how we connect pairs of degree-three vertices.
 Round vertices correspond to the vertices in our graph,
 where filled round vertices are the vertices we want to connect.
 Diamond-shaped vertices correspond to the points in the grid shifted by~$1/3\pm \eps$ horizontally and~$1/2$ vertically.
 Thick lines depict the pairwise connections.}
 \label{fig:grid}
\end{figure}
As we thereby possibly introduce
edge crossings,
we need to \emph{dissolve} them as follows.

\begin{definition}
 Let~$G$ be a graph embedded in the two-dimensional plane
 such that at most two edges cross in one coordinate.
 Let~$e_1,e_2$ be two edges crossing in a coordinate~$(i,j)$.
 \emph{Dissolving} the crossing 
 is doing the following:
 subdivide edge~$e_1$ 
 (denote the vertex $v_1$) 
 and edge~$e_2$
 (denote the vertex $v_2$),
 identify~$v_1$ with~$v_2$
 (denote the vertex~$v$),
 and embed~$v$ at coordinate~$(i,j)$.
\end{definition}

\noindent
The way we dissolve crossings immediately gives the following.

\begin{observation}
 \label{obs:dissolving}
 Every vertex resulting from a dissolution has degree four
 and
 dissolving all edge-crossings yields a planar graph.
\end{observation}

\noindent
We will add and embed edges between disjoint pairs of degree-three vertices,
dissolve each newly formed crossing,
and replace each edge introduced by the dissolution by an~$R$-insertion on its endpoints.
Formally:

\begin{definition}
  \label{def:Rconnect}
  Let~$0<\eps<1/3$.
  $R$-connecting vertex~$v$ with~$v'$,
  where~$p(v)=(i,j)$ and~$p(v')=(i',j')$,
  is doing the following:
  \begin{enumerate}
   \item Add and embed a new edge~$f=\{v,v'\}$ as follows:
    \begin{description}
      \item[if $i=i'$:] It goes from~$(i,j)$ to~$(i-\frac{1}{3}-\eps,j+\frac{1}{2})$
        to~$(i'-\frac{1}{3}-\eps,j'-\frac{1}{2})$ and finally to~$(i',j')$.
      \item[if~$i\neq i'$:] It goes from~$(i,j)$ to~$(i+\frac{1}{3}-\eps,j-\frac{1}{2})$ to~$(i+\frac{1}{3}-\eps,j'-\frac{1}{2})$ to~$(i'-\frac{1}{3}+\eps,j'-\frac{1}{2})$ and finally to~$(i',j')$. 
    \end{description}
    \item Dissolve every crossing,
    and let~$f_1,\dots,f_\ell$ denote the edges in which~$f$ is dissolved. 
    \item Replace each edge~$f_i$,
    $i\in\set{\ell}$,
    with an~$R$-insertion on its endpoints.
  \end{enumerate}
\end{definition}

\noindent
As a technical remark,
we choose~$\eps$ in~\cref{def:Rconnect}
such that no existing slope is resampled.
Note that in our embedding (\cref{thm:gridembedding})
the number of slopes is finite.
Thus,
we can $R$-connect any two vertices in polynomial time.

\begin{algorithm}[t]
$d\setto 0$;
$a\setto \emptyset$; $G'\setto G$\;
\For(\tcp*[f]{x-coordinates}){$x$ from~$1$ to $2n-4$}{
  \For(\tcp*[f]{y-coordinates}){$y$ from~$1$ to~$n-2$}{
    \If{$p^{-1}(x,y)=v$ is a degree-three vertex}{
      \eIf{$a=\emptyset$}{$a\setto v$;}{
      $R$-connect~$a$ with~$v$ in~$G'$ 
      (denote the obtained graph again by~$G'$) 
      and let~$d'$ denote the number of~$R$-insertions\;
      $d\setto d+d'$;
      $a\setto \emptyset$\;
      }
    }
  }
}
\Return{$(G',k+3d)$}
\caption{Computing an equivalent instance~$(G',k)$ with~$G'$ being 4-regular planar from~$(G,k)$ with~$G$ being of minimum degree three and maximum degree four embedded with straight-lines on the~$(2n-4\times n-2)$-grid,
where~$n$ denotes the number of vertices of~$G$. 
}
\label{alg:rchains}
\end{algorithm}
We employ~\cref{alg:rchains}
to construct our instance~$(G',k')$
(see~\cref{fig:deg3exillu} for an illustration).
\begin{figure}[t]
 \centering
 \begin{tikzpicture}

      \def\xr{0.8}
      \def\yr{0.8}
    
    \tikzpramble{}
    
    \newcommand{\ExGrid}{
      \foreach\x in {0,...,3}{
        \foreach \y in {0,...,3}{
        \node (a\x\y) at (\x*\xr,\y*\yr)[xnode]{};
        }
      }
      \draw[xedge] (a00) to (a10) to (a20) to (a30) to (a31) to (a32) to (a22) to (a12) to (a02) to (a01) to (a11);
      \draw[xedge] (a21) to (a31);
      \draw[xedge] (a10) to (a11) to (a30);
      \draw[xedge] (a22) to (a21) to (a12);
      \draw[xedge] (a00) to (a01);
      \draw[xedge] (a02) to (a11);
      \draw[xedge] (a21) to (a32);
      \draw[xedge] (a02) to (a03) to (a13) to (a23) to (a33) to (a32);
      \draw[xedge] (a12) to (a13);
      \def\sc{1.0}
      \draw[-,dashed] (a01) to ($(a01)+(-\sc*\teps,0)$);
      \draw[-,dashed] (a00) to ($(a00)+(-\sc*\teps,0)$);
      \draw[-,dashed] (a03) to ($(a03)+(-\sc*\teps,0)$);
      \draw[-,dashed] (a03) to ($(a03)+(0,\sc*\teps)$);
      \draw[-,dashed] (a23) to ($(a23)+(0,\sc*\teps)$);
      \draw[-,dashed] (a33) to ($(a33)+(0,\sc*\teps)$);
      \draw[-,dashed] (a30) to ($(a30)+(0,-\sc*\teps)$);
      \draw[-,dashed] (a10) to ($(a10)+(0,-\sc*\teps)$);
      \draw[-,dashed] (a20) to ($(a20)+(-\sc*\teps,-\sc*\teps)$);
    }
    
    \newcommand{\markThem}{
      \node at (a00)[xnode,fill=green]{};
      \node at (a13)[xnode,fill=green]{};
      \node at (a20)[xnode,fill=green]{};
      \node at (a22)[xnode,fill=green]{};
      \node at (a23)[xnode,fill=green]{};
      \node at (a30)[xnode,fill=green]{};
      \node at (a31)[xnode,fill=green]{};
      \node at (a33)[xnode,fill=green]{};
    }

    \begin{scope}[xshift=0*\xr cm, yshift=9.5*\yr cm]
      \ExGrid{}
      \markThem{}
      \node at (-0.75*\xr,3.75*\yr)[]{(a)};
    \end{scope}
    
    \begin{scope}[xshift=0*\xr cm, yshift=5*\yr cm]
     \Grid{a}{3}{3}{0}{0}{}{};
     \markThem{}
      
      \Grid{b}{3}{3}{-.66}{-0.5}{xnodey,color=blue}{xedgedot,draw=none};
      \Grid{c}{3}{3}{-.33}{-0.5}{xnodey,color=magenta}{xedgedot,draw=none};
      
      \draw[xpath] (a00) to (b10) to (b13) to (c13) to (a13);
      \draw[xpath] (a20) to (c21) to (c22) to (a22);
      \draw[xpath] (a31) to (c32) to (c33) to (a33);
      \draw[xpath] (a23) to (b33) to (b30) to (c30) to (a30);
      
      \node at (-0.75*\xr,3.75*\yr)[]{(b)};
    \end{scope}

    \begin{scope}[xshift=0*\xr cm, yshift=0*\yr cm]
     \ExGrid{};
     \markThem{}
      
      \Grid{b}{4}{4}{-.66}{-0.5}{inner sep=0pt,draw=none}{xedgedot,draw=none};
      \Grid{c}{4}{4}{-.33}{-0.5}{inner sep=0pt,draw=none}{xedgedot,draw=none};
      
      \draw[xedge,color=blue] (a00) to (b10) to (b13) to (c13) to (a13);
      \draw[xedge,color=blue] (a20) to (c21) to (c22) to (a22);
      \draw[xedge,color=blue] (a31) to (c32) to (c33) to (a33);
      \draw[xedge,color=blue] (a23) to (b33) to (b30) to (c30) to (a30);
      
      \node at (-0.75*\xr,3.75*\yr)[]{(c)};
    \end{scope}

    \newcommand{\coolcon}[3]{%
      \draw[xedge,color=blue] (#1) to [#3]node[inner sep=1pt,midway,fill=white,sloped,scale=\rsc,yshift=\ysh em]{\tikz{\theRcut{0}{0}{0};}}(#2);
    }

    \begin{scope}[xshift=6*\xr cm,yshift=1.5*\yr cm]

        \def\xr{2.6}
        \def\yr{2.6}
      
      \ExGrid{}
      \markThem{}

      \node (d01) at (0.33*\xr,0*\yr)[xnodex]{};
      \node (d02) at (0.33*\xr,1*\yr)[xnodex]{};
      \node (d03) at (0.33*\xr,1.66*\yr)[xnodex]{};
      \node (d04) at (0.33*\xr,2*\yr)[xnodex]{};

      \node (d11) at (2*\xr-0.33*\xr,0.66*\yr)[xnodex]{};
      \node (d12) at (2*\xr-0.33*\xr,1.33*\yr)[xnodex]{};
      
      \node (d21) at (3*\xr-0.33*\xr,1.66*\yr)[xnodex]{};
      \node (d22) at (3*\xr-0.33*\xr,2*\yr)[xnodex]{};
      
      \node (d31) at (2.33*\xr,2*\yr)[xnodex]{};
      \node (d32) at (2.33*\xr,1.33*\yr)[xnodex]{};
      \node (d33) at (2.33*\xr,1*\yr)[xnodex]{};
      \node (d34) at (2.33*\xr,0.33*\yr)[xnodex]{};
      \node (d35) at (2.33*\xr,0*\yr)[xnodex]{};
      
      \def\rsc{0.08}
      \def\ysh{1.5}
      
      \coolcon{a00}{d01}{out=-110,in=-70,looseness=2}
      \coolcon{d01}{d02}{}
      \coolcon{d02}{d03}{}
      \coolcon{d03}{d04}{}
      \coolcon{d04}{a13}{}
      
      \coolcon{a20}{d11}{}
      \coolcon{d11}{d12}{}
      \coolcon{d12}{a22}{}
      
      \coolcon{a31}{d21}{}
      \coolcon{d21}{d22}{}
      \coolcon{d22}{a33}{}
      
      \coolcon{a23}{d31}{}
      \coolcon{d31}{d32}{}
      \coolcon{d32}{d33}{}
      \coolcon{d33}{d34}{}
      \coolcon{d34}{d35}{}
      \coolcon{d35}{a30}{out=-45,in=-135,looseness=1.5}

      \node at (-0.25*\xr,3.175*\yr)[]{(d)};
    \end{scope}
  
 \end{tikzpicture}
 \caption{Illustration to the proof of~\cref{prop:34to4}.
 (a) An example graph~$G$ (excerpt) embedded on a grid where filled vertices correspond to vertices to connect (vertices of degree three).
 (b) The connecting paths constructed in the grid.
 (c) Embedding the edges in~$G$.
 (d) Dissolving the edge-crossings and replacing edges by~$R$-insertions.
 }
 \label{fig:deg3exillu}
\end{figure}
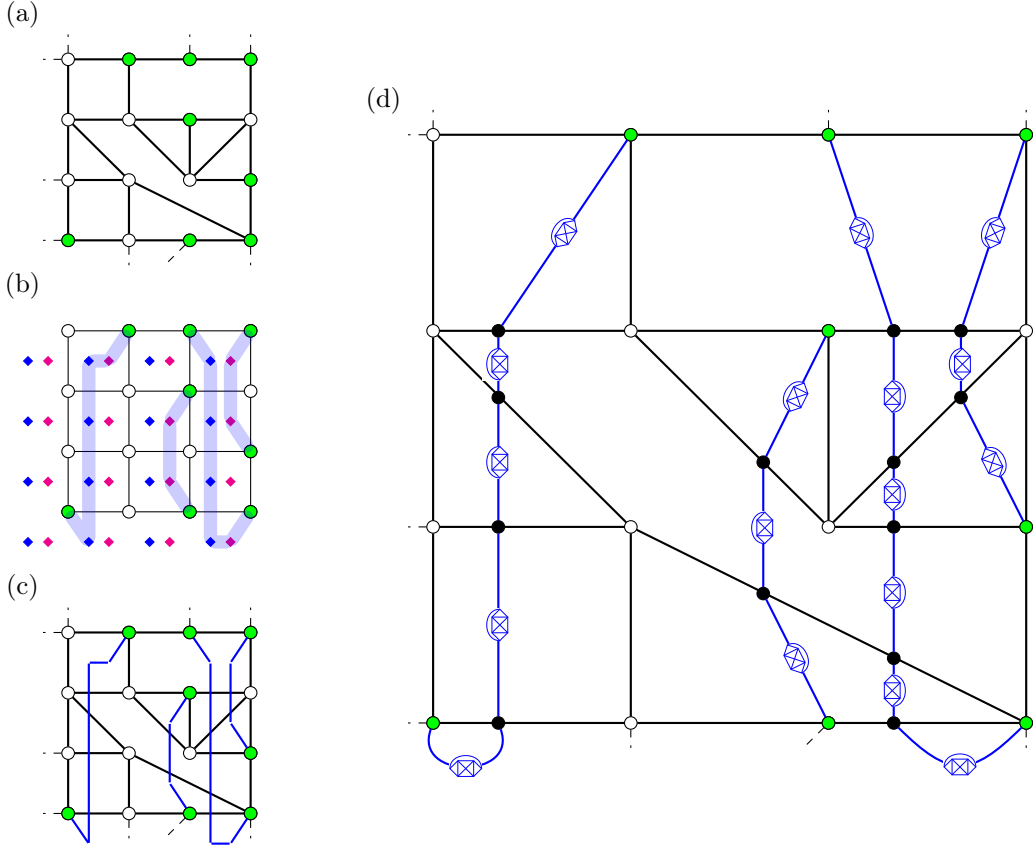
The following invariant immediately holds for~\cref{alg:rchains} 
by our ``left-to-right bottom-to-top'' approach
in an embedding given by~\cref{thm:gridembedding}.

\begin{observation}
 When~\cref{alg:rchains} detects 
 two degree-three vertices~$v$ (first) 
 and~$v'$ (second)
 with~$p(v)=(i,j)$ and~$p(v')=(i',j')$ to connect, 
 then,
 \begin{description}
  \item[if~$i=i'$:]
  there is and was no degree-three vertex~$w$, $p(w)=(x,y)$, 
  with~$x=i$ and~$j<y<j'$;
  \item[if~$i\neq i'$:]
  there is and was no degree-three vertex~$w$, $p(w)=(x,y)$, 
  with
  (i)~$i<x<i'$,
  (ii)~$x=i$ and~$y>j$,
  or
  (iii)~$x=i'$ and~$y<j'$.
 \end{description}
\end{observation}

\noindent
It follows that every two~$R$-connections will be non-crossing and vertex-disjoint.

\begin{proposition}%
 \label{prop:34to4}
 Let~$\I=(G,k)$ be an instance of~\fvsTsc{} with
 $G$ being planar and of minimum degree three and maximum degree four.
 Then one can compute an equivalent instance~$\I'=(G',k')$
 with~$k'\in \O(|V(G)|^2)$ and~$G'$ having~$\O(|V(G)|^2)$ vertices and edges,
 and being 4-regular planar.
\end{proposition}

\begin{proof}
 Let~$\I=(G,k)$ be an instance of \fvsAcr{} where~$G$ is planar and of minimum degree three and maximum degree four.
 Let~$n\ceq |V(G)|$.
 Compute an embedding on the~$(2n-4)\times (n-2)$-grid in polynomial time
 (\cref{thm:gridembedding}).
 Employ \cref{alg:rchains}
 to obtain instance~$\I'=(G',k')$ in polynomial time.
 Since~$G'$ is obtained by a series of subdivisions and~$R$-insertions,
 where we count additional three vertices to delete for every~$R$-insertion,
 $\I'$ is equivalent to~$\I$ (\cref{obs:RinsPlus3}).
 Moreover,
 since we dissolve every crossing,
 $G'$~is planar and 4-regular~(\cref{obs:dissolving}).
 
 Each pair 
 of the at most~$\ceil{n/2}$ pairs 
 crosses at most each of the at most~$3n-6$ edges
 twice.
 Since each two of the $R$-chains are disjoint,
 there are at most~$\O(n^2)$ subdivisions and~$R$-insertions.
 \lqed
\end{proof}

\noindent
We are set to prove 
\cref{thm:fvs:4regplanar}.

\begin{proof}[Proof of~\cref{thm:fvs:4regplanar}]
  Let~$(G,k)$ be an instance of the 
  \NP-hard
  \fvsTsc{} on planar graphs of minimum degree three and maximum degree four
  (\cref{prop:fvs:threefour}).
  Due to~\cref{prop:34to4},
  we can obtain an equivalent instance~$(G',k')$ 
  in polynomial~time where
  $G'$ is planar and 4-regular.
  \lqed
\end{proof}
\subsection{4-regular planar Hamiltonian}
\label{ssec:4regplanarham}

In~\cref{ssec:4regplanar},
we proved that
\fvsTsc{} is \NP-hard on 4-regular planar graphs.
We next give a polynomial-time many-one reduction 
to an equivalent instance with a 4-regular planar Hamiltonian graph.

\begin{theorem}%
 \label{thm:fvs:4regplanaHam}
 \fvsTsc{} on 4-regular planar Hamiltonian graphs
 is \NP-hard,
 even if a Hamiltonian cycle is provided.
\end{theorem}

\noindent
We will follow the idea of Fleischner and Sabidussi~\cite{FleischnerS03}:
We first compute a 2-factor in polynomial time,
and then iteratively connect cycles from the 2-factor by $L$-insertions
to obtain a Hamiltonian cycle.
Note that~$L$ contains 
two vertex-disjoint~$C_4*K_1$s,
hence admits no feedback vertex set of size three
and no feedback vertex set of size four containing~$x$ or~$y$
(\cref{obs:c4k1}).
Yet,
$L$~admits feedback vertex sets each of size four disconnecting~$x$ and~$y$
(see~\cref{fig:gtb}).

\begin{lemma}
 \label{obs:L}
 Graph~$L$ is planar,
 admits a Hamiltonian~$x$-$y$~path,
 has no feedback vertex set of size three,
 but one of size four disconnecting~$x$ and~$y$,
 and none of size four containing one of~$x$ or~$y$.
\end{lemma}

\noindent
\cref{obs:L} immediately implies the following.

\begin{observation}
 \label{lem:Linsertion}
 Let~$\I=(G,k)$ be an instance of~\fvsTsc{}
 and let~$u,v\in V(G)$.
 Let~$G'$ be the graph obtained from~$G$ by an~$L$-insertion at~$u,v$
 and let~$k'\ceq k+4$.
 Then~$\I$ is a \yes-instance
 if and only if
 $(G',k')$  is a \yes-instance of~\fvsTsc{}.
\end{observation}

\noindent
A \emph{2-factor} in a graph~$G$ is a spanning 2-regular subgraph.
We represent a 2-factor by its components~$Q=\{Q_1,\dots,Q_q\}$,
where~$Q_i$ is a cycle for every~$i\in\set{q}$.
Every 4-regular graph admits a 2-factor computable in polynomial time~\cite{Mulder92,FleischnerS03}.
We first compute a 2-factor~$Q=\{Q_1,\dots,Q_q\}$ of~$G$,
and then iteratively 
make~$L$-insertions to merge cycles from~$Q$
until~$|Q|=1$.
\begin{construction}
 \label{constr:2factormerging}
 Let~$G$ be a connected planar 4-regular graph, 
 and let~$Q$ be a 2-factor of~$G$.
 If~$|Q|>1$,
 then there are two distinct cycles~$Q_i,Q_j\in Q$
 with two adjacent vertices~$u,v$ with~$u\in V(Q_i)$ and~$v\in V(Q_j)$.
 We distinguish two cases
 (see~\cref{fig:2factor} for an illustration).
 \begin{figure}[t]
  \centering
  \begin{tikzpicture}
   \def\xr{1}
   \def\yr{1}
   \def\xsh{2.625}
   \tikzpramble{}
   
   \newcommand{\exmex}{%
    \node (u) at (0,0)[xnode,label=90:{$u$}]{};
    \node (v) at (\xsh*\xr,0)[xnode,label=90:{$v$}]{};
    \node (u1) at (-1*\xr,0.5*\yr)[]{};
    \node (u2) at (-1*\xr,0*\yr)[]{};
    \node (u3) at (0*\xr,-1*\yr)[]{};
    \node (v1) at (\xsh*\xr+1*\xr,0.5*\yr)[]{};
    \node (v2) at (\xsh*\xr+1*\xr,0*\yr)[]{};
    \node (v3) at (\xsh*\xr,-1*\yr)[]{};
    \draw[xedge] (u) to (v);
    \foreach\x in {1,...,3}{\draw[xedge] (u) to (u\x);\draw[xedge] (v) to (v\x);}
   }
   \newcommand{\exmexx}{%
    \node (u) at (0,0)[xnode,label=90:{$u$}]{};
    \node (v) at (\xsh*\xr,0)[xnode,label=90:{$v$}]{};
    \node (u1) at (-1*\xr,0.5*\yr)[]{};
    \node (u2) at (-1*\xr,0*\yr)[]{};
    \node (u3) at (0*\xr,-1*\yr)[]{};
    \node (v1) at (\xsh*\xr+1*\xr,0.5*\yr)[]{};
    \node (v2) at (\xsh*\xr+2*\xr,-0.6*\yr)[]{};
    \node (v3) at (\xsh*\xr,-1.5*\yr)[]{};
    \draw[xedge] (u) to (v);
    \foreach\x in {1,...,3}{\draw[xedge] (u) to (u\x);}
    \foreach\x in {1,3}{\draw[xedge] (v) to (v\x);}
    \draw[xedge] (v) to [out=15,in=90](v2);
   }
   
   \begin{scope}[]
    \node at (-1.25*\xr,0.75*\yr)[]{(a)};
    \exmex{}
    \draw[xpathx] (u1) to (u) to (u3);
    \draw[xpathy] (v2) to (v) to (v3);
   \end{scope}
   
   \begin{scope}[yshift=-2.25*\yr cm]
    
    \node at (\xsh*0.5,0.75*\yr)[rotate=-90]{$\leadsto$};
    \exmex{}
    \begin{scope}[scale=0.75]
    \theL{a}{b}{1}{-0.75}
    
    \draw[xpath] (u1) to (u) to (v) to (v2);
    \draw[xpath] (ybx) to (v3);
    \draw[xpath] (xax) to (u3);
    \theLpath{a}{b}
    \end{scope}
   \end{scope}
   
   \begin{scope}[xshift=7*\xr cm]
    \node at (-1.25*\xr,0.75*\yr)[]{(b)};
    \exmexx{}
    \draw[xpathx] (u2) to (u) to (u3);
    \draw[xpathy] (v1) to (v) to [out=15,in=90](v2);
   \end{scope}
   
   \begin{scope}[xshift=7*\xr cm,yshift=-2.25*\yr cm]
    \node at (\xsh*0.5,0.75*\yr)[rotate=-90]{$\leadsto$};
    \exmexx{}
    \begin{scope}[scale=0.75]
    \theL{a}{b}{1}{-0.75}
    \theLpath{a}{b}
    \end{scope}
    
    \begin{scope}[scale=0.625,shift={(5.1*\xr,-1.5*\yr)},rotate=30]
    \theL{c}{d}{0}{0}
    \theLpath{c}{d}
    \end{scope}
    \draw[xpath] (u2) to (u) to (v) to (v1);
    \draw[xpath] (ybx) to (xcx);
    \draw[xpath] (xax) to (u3);
    \draw[xpath] (ydx) to (v2);
   \end{scope}
  \end{tikzpicture}
  \caption{Illustration to~\cref{constr:2factormerging} of (a) Case 1 and (b) Case 2.
  Indicated are
  two different 2-factor components in each upper part (magenta and cyan)
  and the ``merged'' cycle,
  being part of the new 2-factor with one less component,
  in the lower part (blue).}
  \label{fig:2factor}
 \end{figure}
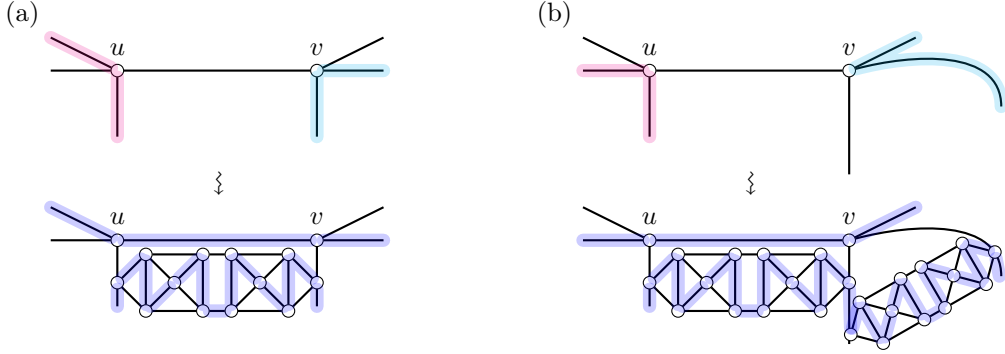

 \begin{description}
  \item[Case 1:] There are two edges~$e\in E(Q_i)$ and~$e'\in E(Q_j)$
 with~$u\in e$ and~$v\in e'$
 sharing the same face.
 Then subdivide~$e$ (denote the obtained vertex~$z$),
 subdivide~$e'$ (denote the obtained vertex~$z'$),
 and then make an~$L$-insertion at~$z,z'$
 (see~\cref{fig:2factor}(a)).
 Since~$L$ admits a Hamiltonian path starting at~$x$ and ending at~$y$
 (\cref{obs:L}, see~\cref{fig:gtb}(b)),
 we can construct a new 2-factor merging~$Q_i$ and~$Q_j$ using the newly added vertices and the edge~$\{u,v\}$.
  \item[Case 2:] 
 There are no two edges~$e\in E(Q_i)$ and~$e'\in E(Q_j)$
 with~$u\in e$ and~$v\in e'$
 sharing the same face.
 Since every vertex is of degree four,
 there is an edge~$e\in E(Q_i)$ with~$u\in e$ that shares a face with an edge~$\tilde{e}\notin E(Q_i)$ with~$u\in \tilde{e}$ sharing a face with an edge
 $e'\in E(Q_j)$
 with~$v\in e'$.
 Subdivide edge~$e$ (denote the obtained vertex~$z$),
 the edge~$\tilde{e}$ twice (denote the obtained vertices~$\tilde{z}_1,\tilde{z}_2$),
 and the edge~$e'$ (denote the obtained vertex~$z'$).
 Now,
 make two~$L$-insertions,
 one at~$z,\tilde{z}_1$ and one at~$\tilde{z}_2,z'$
 (see~\cref{fig:2factor}(b)).
 Again (\cref{obs:L}, see~\cref{fig:gtb}(b)),
 we can construct a new 2-factor merging~$Q_i$ and~$Q_j$ using the newly added vertices and the edge~$\{u,v\}$.\cqed
 \end{description}
\end{construction}

\noindent
We get the following.

\begin{lemma}
 Let~$\I=(G,k)$ be an instance of~\fvsTsc{}
 with~$G$ being connected 4-regular planar
 and let~$Q$ be a 2-factor of~$G$ with~$|Q|\geq 2$.
 Let~$G'$ and~$Q'$ be the graph and the 2-factor obtained from~\cref{constr:2factormerging},
 respectively,
 and let~$k'\ceq k+4$ (in Case~1)
 or~$k'\ceq k+8$ (in Case~2).
 Then~$\I$ is a \yes-instance
 if and only if
 $(G',k')$ is a \yes-instance of~\fvsTsc{}.
 Moreover,
 $G'$ is connected planar 4-regular
 and~$|Q'|=|Q|-1$.
\end{lemma}

\noindent
We are set to prove 
the main result of this section.
\begin{proof}[Proof of~\cref{thm:fvs:4regplanaHam}]
 Let~$\I=(G,k)$ be an instance of \fvsTsc{}
 on connected 4-regular planar graphs.
 Compute a 2-factor~$Q$ of~$G$ in polynomial time.
 Next,
 iteratively apply~\cref{constr:2factormerging} 
 (at most~$n/3$ times~\cite{FleischnerS03})
 to obtain an equivalent instance~$\I'=(G',k',C')$
 in polynomial time,
 where
 $G'$ is 4-regular planar Hamiltonian graph
 with Hamiltonian cycle~$C'$.
 \lqed
\end{proof}
\subsection{5-regular planar Hamiltonian}
\label{ssec:5regplanarham}

In this section,
we prove that~\fvsTsc{}
is also \NP-hard on 5-regular planar Hamiltonian graphs with provided Hamiltonian cycle.
To this end, 
we start from a 4-regular planar Hamiltonian graph
and then make it 5-regular by a~$D$-insertion at every second edge of a Hamiltonian cycle.
Hence,
we need to
ensure the 4-regular Hamiltonian graph to have an even number of vertices.
To this end,
we take two disjoint copies of the input graph
and connect them via two~$L$-insertions
(see~\cref{fig:4regeven} for an illustration).
\begin{figure}[t]
  \centering
  \begin{tikzpicture}
    \def\xr{1}
    \def\yr{1}
    \tikzpramble{};
    
    \theL{a1}{a2}{0}{0.75}
    \theL{b1}{b2}{0}{-0.75}
    
    \node at (xa1x)[label=180:{$v$}]{};
    \node at (ya2x)[label=0:{$v'$}]{};
    \node at (xb1x)[label=180:{$w$}]{};
    \node at (yb2x)[label=0:{$w'$}]{}; 
    
    \node (a1) at (-2*\xr,1.5*\yr)[xnode]{};
    \node (a2) at (-2*\xr,-1.5*\yr)[xnode]{};
    \draw[xedge] (a1) to (xa1x) to (xb1x) to (a2);
    \node (b1) at (3.5*\xr,1.5*\yr)[xnode]{};
    \node (b2) at (3.5*\xr,-1.5*\yr)[xnode]{};
    \draw[xedge] (b1) to (ya2x) to (yb2x) to (b2);
    
    \node[left =of a1] (a11)[]{};
    \node[below left=of a1] (a12)[]{};
    \node[below =of a1] (a13)[]{};
    \foreach\x in{1,2,3}{\draw[-,dashed] (a1) to (a1\x);}
    
    \node[left =of a2] (a21)[]{};
    \node[above left=of a2] (a22)[]{};
    \node[above =of a2] (a23)[]{};
    \foreach\x in{1,2,3}{\draw[-,dashed] (a2) to (a2\x);}
    
    \node[right =of b1] (b11)[]{};
    \node[below right=of b1] (b12)[]{};
    \node[below =of b1] (b13)[]{};
    \foreach\x in{1,2,3}{\draw[-,dashed] (b1) to (b1\x);}
    
    \node[right =of b2] (b21)[]{};
    \node[above right=of b2] (b22)[]{};
    \node[above =of b2] (b23)[]{};
    \foreach\x in{1,2,3}{\draw[-,dashed] (b2) to (b2\x);}
    
    \draw[xpathx] ($(a12) + (\teps,0)$) to ($(a1)+(0,-\teps)$) to ($(xa1x)+(0,-\teps)$) to ($(xb1x)+(0,\teps)$) to ($(a2)+(0,\teps)$) to ($(a21)+(0,\teps)$);
    \draw[xpathy] ($(b12) + (-\teps,0)$) to ($(b1)+(0,-\teps)$) to ($(ya2x)+(0,-\teps)$) to ($(yb2x)+(0,\teps)$) to ($(b2)+(0,\teps)$) to ($(b21)+(0,\teps)$);
    \theLpath{a1}{a2}
    \theLpath{b1}{b2}
    \draw[xpath] (a12) to (a1) to (xa1x);
    \draw[xpath] (a21) to (a2) to (xb1x);
    \draw[xpath] (b12) to (b1) to (ya2x);
    \draw[xpath] (b21) to (b2) to (yb2x);
    
    \draw[dotted] (-2.75*\xr,0*\yr) circle (1.9cm);
    \draw[dotted] (4.25*\xr,0*\yr) circle (1.9cm);
    \node at (-2.75*\xr,0*\yr)[]{$G$};
    \node at (4.25*\xr,0*\yr)[]{$G'$};

  \end{tikzpicture}
  \caption{Illustration to the proof of~\cref{lem:4regeven}.
  Hamiltonian cycles~$C$ and~$C'$ of~$G$ and~$G'$ are depicted on the left and right in magenta and cyan, respectively.
  Hamiltonian cycle~$C^*$ obtained after the~$L$-insertions at~$v,v'$ and~$w,w'$ is depicted in blue.
  }
  \label{fig:4regeven}
\end{figure}
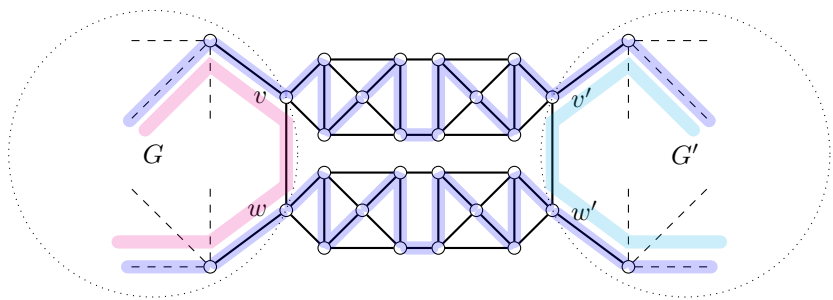

\begin{lemma}%
 \label{lem:4regeven}
 \fvsTsc{} is \NP-hard on 4-regular planar Hamiltonian graphs with an even number of vertices,
 even if an Hamiltonian cycle is provided.
\end{lemma}

\begin{proof}
  Let~$\I=(G,k,C)$ be an instance of \fvsTsc{} on 4-regular planar Hamiltonian graphs with Hamiltonian cycle~$C$.
  If~$V(G)$ is even,
  then we are done.
  Otherwise,
  do the following
  (see~\cref{fig:4regeven} for an illustration).

  Let~$G^*$ denote the graph obtained from taking~$G$ 
  and a disjoint copy~$G'$ of~$G$.
  Note that~$V(G^*)$ is even,
  and~$(G^*,2k)$ is equivalent to~$\I$.
  Let~$C'$ denote the copy of~$C$ in~$G'$,
  and let for some edge~$e\in C$ edge~$e'\in C'$ denote the copy of~$e$.
  We know that we can subdivide~$e,e'$ without changing the feedback vertex set,
  so doubly subdivide edge~$e$ (call the vertices~$v,w$)
  and edge~$e'$ (call the vertices~$v',w'$).
  Next,
  make an~$L$-insertion at~$v,v'$ and one at~$w,w'$.
  Note that we added~$24$ vertices,
  and thus~$V(G^*)$ is even.
  The instance~$(G^*,2k+8)$ is equivalent to~$\I$ (\cref{lem:Linsertion}).
  Now,
  observe that by the choices of~$e,e'$,
  we can merge~$C$ and~$C'$ to~$C^*$.
  Hence,
  the instance~$(G^*,2k+8,C^*)$ is equivalent to~$\I$.
  Note that~$G^*$ remains planar:
  we can disjointly embed~$G$ and~$G'$ such that edges~$e$ and~$e'$ of~$C$ and~$C'$
  are on the outer face~\cite{Diestel10,balakrishnan2012textbook}.
  \lqed
\end{proof}

\noindent
Using~\cref{lem:4regeven},
we will prove next the following main result of this section.

\begin{theorem}%
 \label{thm:fvs:5regplanham}
 \fvsTsc{} is \NP-hard on 5-regular planar Hamiltonian graphs
 even if a Hamiltonian cycle is given.
\end{theorem}

\noindent
We will perform a series of~$D$-insertion
(see~\cref{fig:dinsertion} for an illustration).
\begin{figure}[t]
  \centering
  \begin{tikzpicture}
   \def\xr{1}
   \def\yr{1}
   \tikzpramble{}
   
   \begin{scope}[scale=0.75]\theD{1}{0}{0}\end{scope}
   
   \begin{scope}[scale=0.75]\theD{2}{5}{0}\end{scope}
   
   \begin{scope}[scale=0.75]\theD{3}{12}{0}\end{scope}

   \node at (x1x1)[label=-90:{$v_1$}]{};
   \node at (y1x1)[label=-90:{$v_2$}]{};
   \draw[xedge] (x1x1) to (y1x1);
   \node at (x2x1)[label=-90:{$v_3$}]{};
   \draw[xedge] (x2x1) to (y1x1);
   \draw[xpath] (x2x1) to (y1x1);
   \node at (y2x1)[label=-90:{$v_4$}]{};
   \draw[xedge] (x2x1) to (y2x1);

   \node (i1) at (6*\xr,-0.75*0.5*\yr)[xnode]{};
   \draw[xedge] (y2x1) to (i1);
   \draw[xpath] (y2x1) to (i1);
   \node (i2) at (6.375*\xr,-0.75*0.5*\yr)[]{$\cdots$};
   \node (i2) at (6.375*\xr,-0.75*0.5*\yr-0.35*\yr)[]{$\cdots$};
   \node (i2) at (6.75*\xr,-0.75*0.5*\yr)[xnode]{};
   \draw[xedge] (x3x1) to (i2);
   \draw[xpath] (x3x1) to (i2);

   \node at (x3x1)[label=-90:{$v_{n-1}$}]{};
   \node at (y3x1)[label=90:{$v_{n}$}]{};
   \draw[xedge] (x3x1) to (y3x1);

   \node[below of=x1x1](a)[inner sep=0pt]{};
   \node[below of=y3x1](b)[inner sep=0]{};
   
   \def\sh{1.25}
   \def\shx{1.0}
   \draw[xedge,rounded corners] (x1x1) to ($(x1x1)+(0,-\sh*\yr)$) to ($(y3x1)+(0,-\sh*\yr)$)  to (y3x1);
   \draw[xpath,rounded corners] (x1x1) to ($(x1x1)+(0,-\sh*\yr)$) to ($(y3x1)+(0,-\sh*\yr)$)  to (y3x1);
   \draw[xpathx,rounded corners] 
   ($(i2)+(0,-\teps)$) to ($(x3x1)+(0,-\teps)$) to ($(x3x1)+(\teps,0)$) to 
   ($(y3x1)+(-\teps,0)$) to  ($(y3x1)+(-\teps,-\teps)$)
   to ($(y3x1)+(-\teps,-\shx*\yr)$)
   to ($(x1x1)+(\teps,-\shx*\yr)$)
   to ($(x1x1)+(\teps,-\teps)$) to ($(x1x1)+(\teps,0)$) to ($(y1x1)+(-\teps,0)$)  to ($(y1x1)+(0,-\teps)$)
   to ($(x2x1)+(0,-\teps)$) to ($(x2x1)+(\teps,0)$) to ($(y2x1)+(-\teps,0)$) to ($(y2x1)+(0,-\teps)$) to ($(i1)+(0,-\teps)$);
   ;

   \theDpath{1}
   \theDpath{2}
   \theDpath{3}
  \end{tikzpicture}
  \caption{Illustration to the proof of~\cref{thm:fvs:5regplanham}.
  The magenta path depicts the Hamiltonian cycle before the $D$-insertions,
  and the blue path depicts the Hamiltonian cycle after the~$D$-insertions.}
  \label{fig:dinsertion}
 \end{figure}
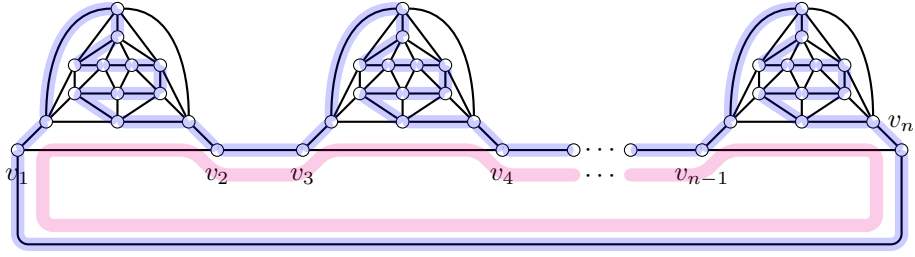
Hence,
we discuss the following in advance.

\begin{lemma}%
 \label{obs:fvs:5regplanham}
 Graph~$D$ is planar,
 admits a Hamiltonian~$x$-$y$~path,
 and has no feedback vertex set of size at most~five
 yet one of size~six containing~$x'$ or~$y'$
 but none of size~six containing~$x$ or~$y$.
\end{lemma}

\begin{proof}
 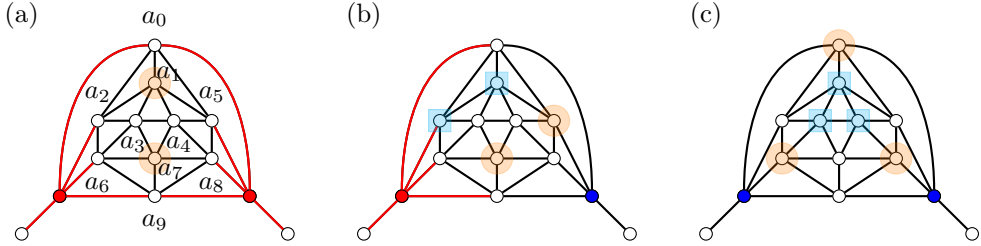
\begin{figure}[t]
  \centering
  \begin{tikzpicture}
   \def\xr{1}
   \def\yr{1}
   \tikzpramble{}
   
   \theD{a}{0}{0}

   \node at (aax1)[xnode,fill=red]{};
   \foreach\x in {xax1,cax1,cax2,cax3}{\draw[xedge,color=red] (\x) -- (aax1);}
   \draw[xedge,color=red] (aax1) to [out=90,in=180](aax2);
   \node at (aax3)[xnode,fill=red]{};
   \foreach\x in {yax1,cax1,cax5,cax6}{\draw[xedge,color=red] (\x) -- (aax3);}
   \draw[xedge,color=red] (aax3) to [out=90,in=0](aax2);
   
   \node at (bax1)[xhili]{};
   \node at (cax4)[xhili]{};

   \node at (aax2)[label=90:{$a_0$}]{};
   \node at (cax4)[label={[xshift=5.5pt,yshift=-6.5pt]90:{$a_1$}}]{};
   \node at (cax3)[label=90:{$a_2$}]{};
   \node at (bax2)[label={[xshift=-1pt,yshift=1pt]-90:{$a_3$}}]{};
   \node at (bax3)[label={[xshift=2pt,yshift=1pt]-90:{$a_4$}}]{};
   \node at (cax5)[label=90:{$a_5$}]{};
   \node at (cax2)[label=-90:{$a_6$}]{};
   \node at (bax1)[label={[xshift=6pt,yshift=5pt]-90:{$a_7$}}]{};
   \node at (cax6)[label=-90:{$a_8$}]{};
   \node at (cax1)[label=-90:{$a_9$}]{};
   
   \node at (-1.75*\xr,2.4*\yr)[]{(a)};

   \theD{b}{4.5}{0}
   \node at (abx1)[xnode,fill=red]{};
   \foreach\x in {xbx1,cbx1,cbx2,cbx3}{\draw[xedge,color=red] (\x) -- (abx1);}
   \draw[xedge,color=red] (abx1) to [out=90,in=180](abx2);
   \node at (abx3)[xnode,fill=blue]{};

   \node at (cbx5)[xhili]{};
   \node at (bbx1)[xhili]{};
   \node at (cbx3)[xxhili]{};
   \node at (cbx4)[xxhili]{};

    \node at (-1.75*\xr+4.5*\xr,2.4*\yr)[]{(b)};
   
   \theD{c}{9}{0}
   \node at (acx1)[xnode,fill=blue]{};
   \node at (acx3)[xnode,fill=blue]{};
   
   \node at (ccx6)[xhili]{};
   \node at (bcx2)[xxhili]{};
   \node at (bcx3)[xxhili]{};
   \node at (acx2)[xhili]{};
   \node at (ccx2)[xhili]{};
   \node at (ccx4)[xxhili]{};

   \node at (-1.75*\xr+9*\xr,2.4*\yr)[]{(c)};
  \end{tikzpicture}
  \caption{Illustration to the proof of~\cref{obs:fvs:5regplanham},
  where (a), (b), and (c),
  correspond to Cases 1, 2, and 3,
  respectively.
  Red-marked vertices and edges correspond to deletions,
  and a blue-marked vertex is prohibited from deletion.
  Vertices highlighted by a circle (orange)
  form a maximum set of pairwise non-adjacent vertices of degree five,
  and vertices highlighted by a rectangle (cyan) are the remaining vertices of degree four.}
  \label{fig:Dsix}
 \end{figure}
 That~$D$ is planar,
 admits and Hamiltonian~$x$-$y$~path,
 and has
 a feedback vertex set of size six containing~$x'$ or~$y'$ 
 is depicted in~\cref{fig:gtb}.
 We assume that
 the vertices of~$D$ is the set~$\{x,x',y,y'\}\cup A$, 
 where~$A\ceq \bigcup_{i=0}^9\{a_i\}$,
 defined as in~\cref{fig:Dsix}(a).
 Let~$D'\ceq D-\{x,y\}$,
 $D_x\ceq D'-\{x'\}$,
 $D_y\ceq D'-\{y'\}$,
 and~$D_{xy}\ceq D'-\{x',y'\}$.
 Suppose towards a contradiction that~$D$ admits a feedback vertex set~$F$ of size at most five. %
 We distinguish three cases
 according to the deletion of~$x'$ or~$y'$
 (let~$F_A\ceq F\cap A$):
 \begin{description}
  \item[Case 1:] $|F\cap\{x',y'\}|=2$.
  Since~$|E(D_{xy})|=21$ and~$|V(D_{xy}-T)|=7$ for every~$T\in\binom{A}{3}$,
  $F_A$ must delete at least 15 edges.
  Note that~$a_1$ and~$a_7$ are the only non-adjacent vertices of degree five in~$D_{xy}$.
  Thus~$F_A$ deletes at most~$14$ edges,
  a contradiction.
  \item[Case 2:] $|F\cap\{x',y'\}|=1$. 
  By symmetry,
  let~$F\cap\{x',y'\}=\{x'\}$.
  Since~$|E(D_x)|=25$,
  $|V(D_{x}-T)|=7$ for every~$T\in\binom{A}{4}$,
  $F_A$ must delete at least~$19$ edges.
  Observe that~$S\ceq \{a_1,a_3,a_4,a_5,a_8\}$ are the only vertices of degree five in~$D_x$.
  Moreover,
  observe that every subset~$S'\subseteq S$ with~$|S'|\geq 3$ 
  is not an independent set in~$D_x$.
  Thus,
  at there at most two non-adjacent vertices of degree five in~$D_x$,
  and hence,
  $F_A$ deletes at most 18 edges,
  a contradiction.
  \item[Case 3:] $|F\cap\{x',y'\}|=0$.
  Since~$|E(D')|=29$,
  $|V(D'-T)|=7$ for every~$T\in\binom{A}{5}$,
  $F_A$ must delete at least~$23$ edges.
  This means that there must be at least three pairwise non-adjacent vertices of degree five in~$F_A$.
  Observe that there is no set~$T\in\binom{A}{4}$ of four pairwise non-adjacent vertices of degree five in~$D'$.
  Observe that for every~$T\in\binom{A}{3}$ of pairwise non-adjacent vertices of degree five in~$D'$,
  there are no two non-adjacent vertices of degree four from~$A\setminus T$ in~$D'-T$.
  It follows that~$F_A$ deletes at most~$22$ edges,
  a contradiction.
 \end{description}
 We conclude that~$D$ admits no feedback vertex set of size at most five.
 \lqed
\end{proof}

\noindent
We are set to prove 
\cref{thm:fvs:5regplanham}.

\begin{proof}[Proof of~\cref{thm:fvs:5regplanham}]
 Let~$\I=(G,k,C)$ be an instance of
 \fvsTsc{} with a 4-regular planar Hamiltonian graph~$G$ with Hamiltonian cycle~$C$
 and an even number of vertices.
 Let~$v_1,v_2,\dots,v_n$ be an ordering of the vertices induced by~$C$.
 Make a $D$-insertion at~$v_{2i-1},v_{2i}$ for every~$i\in\set{n/2}$.

 Let~$G'$ denote the obtained graph.
 Note that~$G'$ is planar and 5-regular.
 Moreover,
 $G'$ is Hamiltonian: extend~$C$ to~$C'$ by adding the Hamiltonian paths for each~$D$ inserted as depicted in~\cref{fig:gtb}.
 Let~$\I'=(G',k',C')$ be the obtained instance,
 where~$k'\ceq k+3n$.
 Since each~$D$-insertion needs six additional feedback deletions 
 (\cref{obs:fvs:5regplanham}),
 and since we can assume each to disconnect~$D-\{x,y\}$ from either~$x$ or~$y$,
 the correctness follows.
 \lqed
\end{proof}

\section{Regular Hamiltonian Graphs}
\label{sec:reg}

In this section,
we prove that \fvsTsc{} remains \NP-hard
on $p$-regular Hamiltonian graphs for \emph{every}~$p\geq 4$.

\begin{theorem}%
 \label{thm:fvs:kreg}
 For every~$p\geq 4$,
 \fvsTsc{} on $p$-regular Hamiltonian graphs is \NP-hard,
 even if a Hamiltonian cycle is provided.
\end{theorem}

\noindent
We have seen that~\fvsTsc{} is \NP-hard on 5-regular graphs,
even if a Hamiltonian cycle~$C$ is provided.
Every 5-regular graph has an even number of vertices.
Thus,
we find a perfect matching~$M$ on~$C$.
We will make a~$Y_p$-insertion on every pair of vertices from~$M$.

\begin{lemma}%
 \label{obs:y:fvs}
 Graph~$Y_p$ is planar,
 admits 
 a Hamiltonian path with endpoints~$x$ and~$y$,
 and
 admits no feedback vertex set of size at most~$2(p+1)-5$
 yet one of size~$2(p+1)-4$ containing~$x'$ or~$y'$.
\end{lemma}

\begin{proof}
 That graph~$Y_p$ is planar and admits 
 a Hamiltonian $x$-$y$~path
 is shown in~\cref{fig:gtb}(d).
 For a feedback vertex set,
 in each~$K_{p+1}$ we need to delete~$p-1$ vertices.
 Thus, 
 we need at least~$2(p-1)=2(p+1)-4$ vertices in a feedback vertex set.
 For the second claim,
 observe the following.
 Delete all vertices in~$K_{p+1}$ except for
 two vertices~$v,w$ different from~$x'$,
 and delete all vertices in the other~$K_{p+1}$ different from~$y'$ and the two neighbors~$v',w'$ of~$v,w$ in the second~$K_{p+1}$.
 It is not difficult to see that 
 after the prescribed deletions,
 a~$2K_2+2K_1$ 
 (the disjoint union of two~$K_2$s and two~$K_1$s) 
 remains.
 \lqed
\end{proof}

\noindent
From~\cref{obs:y:fvs} we immediately get the following.

\begin{observation}
 \label{lem:Ykinsert}
 Let~$\I=(G,k)$ be an instance of~\fvsTsc{}
 and let~$v,w\in V(G)$ be two distinct vertices.
 Let~$G'$ be the graph obtained from~$G$ by the~$Y_p$-insertion at~$v,w$,
 and let~$k'\ceq k+2(p+1)-4$.
 Then,
 $\I$ is a \yes-instance
 if and only if
 $(G',k')$ is a \yes-instance of~\fvsTsc{}.
\end{observation}

\noindent
We are set to prove 
\cref{thm:fvs:kreg}.

\begin{proof}[Proof of~\cref{thm:fvs:kreg}]
 Let~$\I=(G,k,C)$ be an instance of~\fvsTsc{}
 where~$G$ is a $p$-regular Hamiltonian graph with Hamiltonian cycle~$C$
 with an even number~$n$ of vertices.
 Let~$v_1,\dots,v_n$ be the order in which~$C$ meets the vertices of~$G$.
 Then,
 for every~$i\in\set{n/2}$,
 make a~$Y_p$-insertion at~$v_{2i-1},v_{2i}$.
 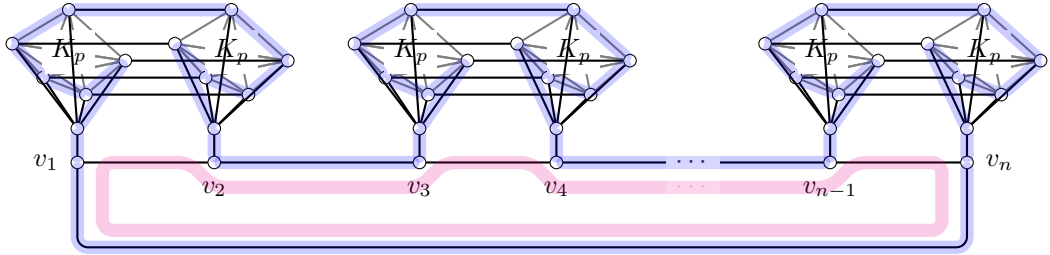
\begin{figure}[t]
  \centering
  \begin{tikzpicture}
    \def\xr{0.9}
    \def\yr{0.9}
    \tikzpramble{}

    \theYk{a}{0}{0}
    \theYkpath{a}
    \node at (xax)[label=180:{$v_1$}]{};
    \node at (yax)[label=-90:{$v_2$}]{};

    \theYk{b}{5}{0}
    \theYkpath{b}
    \node at (xbx)[label=-90:{$v_3$}]{};
    \node at (ybx)[label=-90:{$v_4$}]{};
    
    \theYk{c}{11}{0}
    \theYkpath{c}
    \node at (xcx)[label=-90:{$v_{n-1}$}]{};
    \node at (ycx)[label=0:{$v_{n}$}]{};
    \draw[xedge] (xax) to (yax) to (xbx)  to (ybx)  to node[midway,fill=white]{$\cdots$}(xcx) to (ycx);
    \def\sh{1.25}
    \def\shx{1.0}
    \draw[xedge,rounded corners] (xax) to ($(xax)+(0,-\sh*\yr)$) to ($(ycx)+(0,-\sh*\yr)$)  to (ycx);
    \draw[xpath,rounded corners] (xax) to ($(xax)+(0,-\sh*\yr)$) to ($(ycx)+(0,-\sh*\yr)$)  to (ycx);
    \draw[xpath,rounded corners] (yax) to (xbx);
    \draw[xpath,rounded corners] (ybx) to node[midway,fill=white]{$\cdots$}(xcx);
    
    \draw[xpathx,rounded corners] 
   ($(ybx)+(0,-\teps)$) to node[midway,fill=white]{$\cdots$}($(xcx)+(0,-\teps)$) to ($(xcx)+(\teps,0)$) to 
   ($(ycx)+(-\teps,0)$) to  ($(ycx)+(-\teps,-\teps)$)
   to ($(ycx)+(-\teps,-\shx*\yr)$)
   to ($(xax)+(\teps,-\shx*\yr)$)
   to ($(xax)+(\teps,-\teps)$) to ($(xax)+(\teps,0)$) to ($(yax)+(-\teps,0)$)  to ($(yax)+(0,-\teps)$)
   to ($(xbx)+(0,-\teps)$) to ($(xbx)+(\teps,0)$) to ($(ybx)+(-\teps,0)$) to ($(ybx)+(0,-\teps)$);
  \end{tikzpicture}
  \caption{Illustration to the proof of~\cref{thm:fvs:kreg}.
  The magenta path depicts the Hamiltonian cycle before the $Y_p$-insertions,
  and the blue path depicts the Hamiltonian cycle after the~$Y_p$-insertions.}
  \label{fig:Ykinsertion}
 \end{figure}
 Extend~$C$ to~$C'$ through every inserted~$Y_p$ as depicted in~\cref{fig:Ykinsertion}.
 Set~$k'\ceq k+\frac{n}{2}(2(p+1)-4)$.
 Due to~\cref{lem:Ykinsert},
 the instance~$(G',k',C')$ is equivalent with~$\I$.
 Moreover,
 $G'$ is $(p+1)$-regular.
 \lqed
\end{proof}
\section{Ordered and Connected Hamiltonian Graphs}
\label{sec:ordcon}

For every cycle,
and for every three vertices on it,
we can shift the start of the cycle and its orientation
to encounter the three vertices in any order.
Thus,
we get the following.

\begin{observation}
 \label{fact:ham3ho}
 Every Hamiltonian graph is 3-Hamiltonian-ordered.
\end{observation}

\noindent
We start from here and prove inductively the following.

\begin{theorem}%
 \label{thm:fvs:pho}
 For every $p\geq 3$,
 \fvsTsc{} on $p$-Hamil\-to\-ni\-an-ordered graphs
 is \NP-hard,
 even if a Hamiltonian cycle is provided.
\end{theorem}

\begin{construction}
 \label{constr:fvs:pho}
 Let~$\I=(G,k,C)$ with~$G=(V,E)$ be an input instance of~\fvsTsc{} with Hamiltonian cycle~$C$ of~$G$.
 Let~$n\ceq |V|$.
 We construct an instance~$\I'\ceq (G',k',C')$ with~$k'\ceq 3n+k$ as follows.
 Let~$H\ceq K_{3n}$.
 Construct~$G'\ceq G*H$.
 Moreover,
 add two new vertices~$x,y$ to~$G'$
 and make~$x$ adjacent to all vertices in~$V(H)\cup\{y\}$ and
 $y$ adjacent to all vertices in~$V(H)\cup\{x\}$.
 Extend~$C$ through~$V(H)\cup\{x,y\}$ to obtain~$C'$.
 \cqed
\end{construction}

\begin{observation}%
 \label{obs:fvs:pho}
 Let~$\I'=(G',k',C')$ be the instance
 obtained from an input instance~$\I=(G,k,C)$ of~\fvsTsc{} 
 using~\cref{constr:fvs:pho}.
 If~$G$ is $p$-Hamiltonian-ordered for~$3\leq p\leq n$,
 then~$G'$ is~$(p+1)$-Hamiltonian-ordered.
\end{observation}

\noindent
Our proof of~\cref{obs:fvs:pho} employs the following.

\begin{fact}[\cite{NgS97}]
 \label{fact:degkho}
 Let~$G=(V,E)$ be a graph with~$|V|\geq 3$ and let~$p\in\set[3]{n}$.
 If $\deg(v)+\deg(w)\geq |V|+2p-6$ for every non-adjacent~$v,w\in V$,
 then~$G$ is~$p$-Hamiltonian-ordered.
\end{fact}

\begin{proof}[Proof of~\cref{obs:fvs:pho}]
 Note that the only two non-adjacent vertices~$v,w$ are either 
 (i) both from~$G$ or
 (ii) one from~$G$ and one from~$x,y$.
 In either case,
 we have that~$\deg_{G'}(v)+\deg_{G'}(w)\geq 6n\geq (4n+2)+2(p+1)-6$.
 Thus,
 by~\cref{fact:degkho},
 $G'$ is~$(p+1)$-Hamiltonian-ordered.
 \lqed
\end{proof}

\begin{lemma}%
 \label{lem:fvs:pho}
 Let~$\I'=(G',k',C')$ be the instance
 obtained from an input instance~$\I=(G,k,C)$ of~\fvsTsc{} 
 using~\cref{constr:fvs:pho}.
 Then,
 $\I$ is a \yes-instance 
 if and only if
 $\I'$ is a \yes-instance.
\end{lemma}

\begin{proof}
 \RD{}
 Let~$F\subseteq V(G)$ be a size-$k$ feedback vertex set of~$G$.
 Then,
 $F'\ceq F\cup V(H)$ is a feedback vertex set of~$G'$ of size~$k+3n=k'$.
 
 \LD{}
 Let~$F'\subseteq V(G')$ be a vertex set of~$G'$ of size~$k'$.
 By construction,
 we know that~$|F'\cap (V(H)\cup\{x,y\})|\geq 3n$.
 Let~$F\ceq F'\cap V(G)$.
 We know that~$|F|\leq k$,
 and~$G-F$ is acyclic.
 \lqed
\end{proof}

\noindent
We are set to prove
\cref{thm:fvs:pho}.

\begin{proof}[Proof of~\cref{thm:fvs:pho}]
 We know that \fvsTsc{} is~\NP-hard on Hamiltonian graphs and thus on~$3$-Hamiltonian graphs (\cref{fact:ham3ho})
 with a provided Hamiltonian cycle.
 We apply~\cref{constr:fvs:pho} to obtain from an
 instance~$\I$ with a~$p$-Hamiltonian-ordered graph and a Hamiltonian cycle
 an equivalent (\cref{lem:fvs:pho}) instance with
 a~$(p+1)$-Hamiltonian-ordered graph (\cref{obs:fvs:pho})
 and a Hamiltonian cycle.
 The statement finally follows by induction.
 \lqed
\end{proof}
Recall that
every~$p$-Hamiltonian-ordered graph
is trivially also~$p$-ordered Hamiltonian.
In addition,
the following holds true.

\begin{fact}[\cite{NgS97}]
 If a graph is $p$-ordered for any~$p\geq 3$,
 then it is $(p-1)$-connected.
\end{fact}

\noindent
Hence,
we get the following.

\begin{corollary}
 \label{cor:ordconhard}
 \fvsTsc{} is \NP-hard on~$p$-ordered and $(p-1)$-connected Hamiltonian graphs, 
 $p\geq 3$,
 even if a Hamiltonian cycle is provided.
\end{corollary}

\section{Conclusion}

\fvsTsc{}
remains \NP-hard on quite restricted cases 
even if the graph is additionally Hamiltonian and
a Hamiltonian cycle is provided.
Which problems are
\NP-hard on Hamiltonian graphs 
and
become non-trivially tractable if a Hamiltonian cycle is provided?
Which problems
become non-trivially tractable 
on~$p$-Hamiltonian-ordered graphs?
As to the class of $p$-Hamil\-tonian-ordered graphs,
we are not aware of a computational complexity study
on this class next to ours.
Further,
it is interesting to study
\fvsTsc{} on the intersections of the classes of
regular graphs,
planar graphs,
$p$-Hamiltonian-ordered,
and
$p$-ordered Hamiltonian graphs.

{
\begingroup
  \renewcommand{\url}[1]{\href{#1}{$\ExternalLink$}}
  \newcommand*{\doi}[1]{\href{http://dx.doi.org/#1}{$\ExternalLink$}}
  \bibliography{fvs-ham-arxiv}
\endgroup
}

\end{document}